\documentclass[12pt]{article}

\usepackage[english]{babel}
\usepackage[T1]{fontenc}
\usepackage[utf8]{inputenc}

\usepackage{amssymb,amsfonts,amsmath,amsthm}

\theoremstyle{plain}
\newtheorem{theorem}{Theorem}
\newtheorem{corollary}{Corollary}
\newtheorem{definition}{Definition}
\newtheorem{lemma}{Lemma}
\newtheorem{proposition}{Proposition}

\theoremstyle{remark}

\newtheoremstyle{example}
{1em}                
{2em}                
{\upshape}        
{}                
{\bfseries}       
{.}               
{ }               
{\thmname{#1}\thmnumber{ #2}\thmnote{ (#3)}}                
\theoremstyle{example}
\newtheorem{example}{Example}

\usepackage{dsfont}
\def\1{\mathds{1}}
\let\epsilon\varepsilon
\let\phi\varphi
\usepackage{mathtools}

\usepackage{mleftright}
\mleftright

\usepackage{accents}

\let\widebar\overline

\usepackage{url}

\usepackage[dvipsnames]{xcolor}
\usepackage{graphicx}

\usepackage[onehalfspacing]{setspace}
\usepackage{indentfirst}

\usepackage[shortlabels]{enumitem}

\usepackage{etoolbox}

\usepackage[comma]{natbib}

\usepackage[vcentering,nohead,margin=1.25in]{geometry}

\usepackage{microtype}

\def\de{\mathop{}\!\mathrm{d}}
\def\E{\mathop{}\mathrm{E}}

\def\R{\ensuremath{\mathbf{R}}}
\def\N{\ensuremath{\mathbf{N}}}

\def\I{\ensuremath{\mathbf{i}}}
\newcommand{\df}[1]{\emph{#1}}
\newcommand{\pp}[1]{\left( #1 \right)}
\newcommand{\pb}[1]{\left[ #1 \right]}
\newcommand{\pc}[1]{\left\{ #1 \right\}}
\newcommand{\abs}[1]{\left\lvert #1 \right\rvert}
\newcommand{\g}{\ifnum\currentgrouptype=16 \;\middle|\;\else\mid\fi}

\DeclareMathOperator{\var}{var}

\DeclareMathOperator*{\argmin}{{arg\,min}}
\DeclareMathOperator*{\argmax}{{arg\,max}}

\makeatletter
\newcommand{\Matrix}[1] 
{\begin{pmatrix}
  \Matrix@r #1;\@bye;\Matrix@r
 \end{pmatrix}}
\def\Matrix@r #1;{\@bye #1\Matrix@z\@bye\Matrix@s #1,\@bye, }%
\def\Matrix@s #1,{#1\Matrix@t }%
\def\Matrix@t #1,{\@bye #1\Matrix@y\@bye\@firstofone {&#1}\Matrix@t}%
\def\Matrix@y #1\Matrix@t{\\ \Matrix@r }%
\def\Matrix@z #1\Matrix@r {}
\def\@bye  #1\@bye {}
\makeatother
\def\<#1>{\Matrix{#1}}

\newcommand{\scB}{\mathcal{B}}
\newcommand{\scC}{\mathcal{C}}

\newcommand{\scG}{\mathcal{G}}

\newcommand{\scL}{\mathcal{L}}
\newcommand{\scM}{\mathcal{M}}

\newcommand{\scP}{\mathcal{P}}

\newcommand{\scR}{\mathcal{R}}
\newcommand{\scS}{\mathcal{S}}

\begin{document}

\title{Elicitability\footnote{This version: September 2023. We thank seminar participants at Caltech, the University of Chicago, Cornell, Duke, Harvard/MIT, the University of Iowa, Northwestern, the UK Seminars in Economic Theory, USC, Yale, Zurich, and the audiences of the PARIS Symposium on AI Society and of the Institute for Mathematical and Statistical Innovation.}}
\author{Y. Azrieli\thanks{Dept. of Economics, Ohio State University, \texttt{azrieli.2@osu.edu} and \texttt{healy.52@osu.edu}.} \and C. Chambers\thanks{Dept. of Economics, Georgetown University, \texttt{christopher.chambers@georgetown.edu}.} \and P.J. Healy\footnotemark[1] \and N.S. Lambert\thanks{Dept. of Economics, University of Southern California, \texttt{lambertn@usc.edu}.}}
\date{September 2023}

\maketitle

\begin{abstract}
An analyst is tasked with producing a statistical study.
The analyst is not monitored and is able to manipulate the study. 
He can receive payments contingent on his report and trusted data collected from an independent source, modeled as a statistical experiment.
We describe the information that can be elicited with appropriately shaped incentives, and apply our framework to a variety of common statistical models. 
We then compare experiments based on the information they enable us to elicit. This order is connected to, but different from, the Blackwell order. Data preferred for estimation are also preferred for elicitation, but not conversely. 
Our results shed light on how using data as incentive generator in payment schemes differs from using data for statistical inference.
\end{abstract}

\newpage


\section{Introduction}
\label{sec:introduction}

Imagine a principal who tasks an analyst with producing a statistical study.
The analyst is able to manipulate the study without being detected and, absent additional incentives, has an interest to do so. 
For example, the analyst may work with data that are private or not verifiable, may work with proprietary algorithms, and may have a conflict of interest with the principal.
Suppose the principal has access to independent, trusted data, which may be in too small quantity to replicate the study, but that can be utilized to make contingent payments for the purpose of incentive provision. Our main goal is to understand the extent to which we can incentivize the truthful reporting of information with such data.

To understand the issues, consider the following example. A regulatory agency asks a pharmaceutical company to assess the efficacy of a new drug as part of an evaluation process. This assessment is done by clinical trials to be performed by the firm itself, and could be falsified or fabricated due to a lack of good monitoring practice.\footnote{Although it is hard to detect, there are numerous documented cases of fraud with clinical trials, particularly in cancer medication markets, including fictitious patients, altered laboratory data, fabricated results, or intentional biases in randomization procedures.
See, for example, \citet{george2015data} for a recent account.}
Let the fraction of the population on which the new drug is effective be $\theta$. 
This fraction is what the agency cares to learn, and what the firm has evaluated.
Generally, the firm is not sure about $\theta$. Its knowledge, or belief about $\theta$, takes the form of a probability distribution.
If the agency was eventually able to observe $\theta$ perfectly and in due time, then classical methods could be employed to reward the firm for telling all it knows about $\theta$ and punish every deviation from the truth. It can be done, for example, with the performance score of \citet{matheson1976scoring}.

In the present paper, the relevant `state of nature' is a parameter not directly observed. Instead, the principal collects data that correlate imperfectly with the parameter. 
As a result, the principal may be unable to solicit the full information the analyst possesses about this parameter.
Returning to our example, suppose the agency performs one clinical trial on its own, and by doing so, is able to assess the efficacy of the new drug on just one random individual.
Evidently, this piece of data is insufficient to estimate $\theta$ with reasonable precision. 
It is also not enough to incentivize the firm to reveal all it knows about $\theta$: Payments must depend, at most, on the result of the single trial, and any choice from a menu of such contingent payments does not reveal any information beyond the believed likelihood of a positive trial.
However, it is enough to incentivize the firm to reveal a point estimate of $\theta$. For example, if $x=1$ for a positive trial and $x=0$ for a negative one, and if the firm who reports $\mu$ gets $1 - (\mu - x)^2$, then the firm maximizes expected payoffs exactly when it reports the truth about the assessed mean of $\theta$.

Our first objective is to describe the information for which truthful reports can be incentivized, to an arbitrary magnitude, in a general environment. We call such information \df{elicitable}. The relationship between observations and the parameter of the underlying statistical model is specified by an arbitrary statistical experiment \citep{blackwell1951comparison}. 
Our results show that the elicitable information varies widely as a function of the amount of data.
In our example, one clinical trial enables the agency to elicit the mean of $\theta$, and with two trials, it also elicits the variance. However, the mode and median, that are other typical point estimates, cannot be elicited with any finite number of trials. 
The elicitable information also depends critically on the statistical model.
Suppose that the outcome of a clinical trial is not binary but real valued, that we account for certain covariates such as age and weight, and that we postulate a Gaussian linear model that links covariates to outcomes.
Then, we can incentivize the reporting of the full information by adequately creating payments contingent on a single observation. This ability is lost if we work with a semiparametric linear model that does not specify the shape of the error term. Nonetheless, in this case, with the data contained in four observations, we can still incentivize the reporting of the mean coefficients and their variance-covariance matrix. In addition, the data requirements do not depend on the number of covariates in the model.

Our second objective is to compare datasets---specifically, the experiments that produce them---and investigate the conceptual differences between working with data to estimate model parameters, and working with data to create incentives that elicit information.
The Blackwell order is a classical benchmark, whereby an experiment A is more informative than an experiment B if B is a garbling of A. 
This means that posteriors obtained from the data A supplies are mean-preserving spreads of posteriors obtained from the data B supplies.
The Blackwell order is the relevant mean of comparison when the experiment is carried out by a statistician who estimates model parameters and wants to minimize risks \citep{le1996comparison}.
On the other hand, in an elicitation context, A is preferred to B when the data it produces makes it possible to elicit least as much information than the data issued from B.
Although the elicitation order and the Blackwell order apply to different contexts, our results suggest they are closely related: A preference for A over B in the sense of Blackwell implies the same preference in the sense of elicitation. The converse is false, because the elicitation order makes it possible to do many more comparisons than the Blackwell order: Some type of noise in the data does not affect our ability to offer incentives, while it always impedes learning.
In addition, the power of incentives can be maintained when switching from B to A:
We demonstrate that the elicitation order coincides with an `incentive order' that ranks A above B when the incentives that can be implemented with B can also be implemented with A.
However, to do so, it may be necessary to expand the range of possible payments.
If the range of possible payments is constrained to be nonnegative or bounded within a given interval, then the resulting incentive order is stronger than than the elicitation order but remains weaker than the Blackwell order.

The paper proceeds as follows. 
The remainder of this section discusses the related literature. 
Section~\ref{sec:model} presents the model. 
Section~\ref{sec:elicitable information} provides three main results to describe the information that can be elicited when the principal has access to some given arbitrary experiment, and puts these results to work in several common statistical models.
Section~\ref{sec:comparison of experiments} investigates the comparison of experiments from the viewpoint of information elicitation and contrasts it with the classical Blackwell comparison of experiments.
The appendices include the proofs omitted from the main text.

\subsection*{Related Literature}

To understand the connection with the literature, it is helpful to consider the following framework. 
An agent has a private type $t \in T$. A principal observes a signal $s \in S$ about the agent's type.
The type and signal spaces, $T$ and $S$, are arbitrary, they depend on the context.
The joint distribution over types and signals is common knowledge.
The agent is first asked to report his type, then the signal realizes and the principal makes a transfer to the agent, $\phi(\hat t, s)$, as a function of the report $\hat t$ and the signal $s$. 
In incentive compatible transfer schemes, the agent, assumed to be risk neutral, is best off reporting honestly.
In the abstract, this framework fits many models of the literature.
Our model is also an instance of this framework, whereby the agent is the analyst and his type is the analyst's assessment of the parameter of interest, while the signal represents the data collected by the principal.

The voluminous literature on belief and probability elicitation goes back to \citet{brier1950verification}, \citet{good1952rational}, \citet{mccarthy1956measures}, \citet{definetti1962} and \citet{savage1971elicitation} (for a recent survey, see \citealp{gneiting2007strictly}).
As in our work, the problem being dealt with can be viewed as an instance of the framework above, where the agent's type is a probability distribution over states, and the principal's signal is a state realization drawn according to the agent's type. 
The main focus of this literature is the design of transfer schemes that offer strict incentives to report truthfully.
The fundamental difference with our work is that, in our setup, the parameter of the underlying statistical model, which is the analog of the state in the probability elicitation literature, is not known and cannot be inferred from the principal's signal. For this reason, it is often not possible to find a scheme that offer strict incentives for truthfully reporting the agent's type.
A recent strand of literature simplifies communication by restricting the set of feasible reports.\footnote{See, in particular, \citet{lambert2008eliciting}, \citet{gneiting2011}, \citet{abernethy2012}, \citet{frongillo2015a}.} 
In the framework above, this is saying that the message space used to report the type is smaller than the entire type space. 
In this case, it is not always possible to implement strict incentives for honest reports.
These works examine the message spaces for which transfer schemes with strict incentives continue to exist and study these schemes. 
However, they continue to assume that the full state is observed. The inability to design adequate incentives stems from the impossibility for the agent to express his belief rather than the imperfect knowledge of the principal. Notably, the concept of elicitation in these works differ from ours because we do not constrain the agent's communication.
Other works deal with groups of individuals, as in \citet{miller2005eliciting}. The object of interest is a random signal privately observed by each member. Incentives are provided by exploiting a form of correlation among signals. Once the  correlation structure becomes known, these mechanisms reduce to standard probability elicitation methods.

In the mechanism design paradigm, belief elicitation is central to the problem of surplus extraction.
In a seminal paper, \citet{cremer1988full} provide necessary and sufficient conditions under which an auction can be designed that extracts the full surplus from the bidders when their valuations are correlated in the private value environment.
The auction of \citeauthor{cremer1988full} is composed of a belief elicitation part, where each bidder reports, indirectly through her bid, her belief on the distribution over other bidders' valuations, and an allocation part.
The key feature is the control of the payments in the belief elicitation part, a point explicitly made by \citet{mcafee1992correlated}, who interpret the problem of surplus extraction in general mechanisms as an instance of the framework above.\footnote{For recent developments on the possibility or impossibility of full surplus extraction, see \citet{neeman2004relevance}, \citet{heifetz2006generic}, \citet{barelli2009genericity}, \citet{chen2013genericity}, \citet{rahman2012surplus}, and \citet{lopomo2019detectability}. Notably, \citet{fu2021full} extend the original setting of \citeauthor{cremer1988full} to allow the seller to receive signals on the bidders' valuations and remark that these signals are used to construct contingent payments and not for statistical inference.}
The agent's type is the characteristic of a participant to some mechanism, and the principal's signal can be the characteristic of some other participant assumed to have communicated truthfully.
\citeauthor{mcafee1992correlated} argue that the property needed for full surplus extraction is much stronger than the strict incentive compatibility called for in the probability elicitation literature.%
\footnote{\citeauthor{mcafee1992correlated} argue that the problem of full surplus extraction reduces to showing that any real function of the agent's type matches the ex-interim expected transfer of some incentive compatible transfer scheme. 
It is the reason why \citeauthor{cremer1988full} did not use a standard scheme such as a classical quadratic scoring rule, because these schemes offer very limited control on expected transfers.
When the property identified by \citeauthor{mcafee1992correlated} is satisfied, a mechanism designer can extract the full rent of the participant by setting a participation fee that exactly offsets his gains. This participation fee is set by a transfer scheme whose ex-interim expected transfers are equal to the negative of the participant's profit function in the original mechanism.
}
We differ in our objective and in the richness of our type and signal spaces (this fact is not merely technical but also substantive). Most importantly, the conditions on the joint distribution of $(t,s)$ that are necessary for full surplus extraction are not satisfied in the context of our paper.
Indeed, this condition asks that the family of conditional distributions over signals, given the type, \emph{excludes} all convex combinations. On the contrary, in our work this family \emph{includes} all convex combinations. 
Exporting the conditions of \citeauthor{cremer1988full} or \citeauthor{mcafee1992correlated} to our setup is making the assumption that the principal can always elicit the full information from the analyst---which in most environments mean the principal eventually observes the model parameter---and also rules out the possibility that the analyst may have varying degrees of uncertainty about the parameter---essentially, it demands that the analyst knows the parameter of the underlying statistical model exactly. The latter observation relates to the impossibility of extracting full surplus in auctions whose bidders receive signals of varying informativeness, as demonstrated by \citet{strulovici2017impossibility}.

Finally, another line of research examines the possibility to test the knowledge of forecasters (see, in particular, \citealp{dawid1982,foster1998,olszewski2008,shmaya2008}).
A state---typically, an infinite stream of binary data---is drawn from an unknown data generating process, and a forecaster communicates some data generating process. Here the goal is not to provide incentives but rather to design, prove or disprove the existence of statistical tests that, based on the state realization, distinguish the forecaster who knows the true data generating process from those who don't. 
In most applications we deal with, the statistical model is identified and generates i.i.d. samples, which makes it possible to infer exactly the data generating process from an infinite data stream. \citet{al2010} provide positive results for the case of non-i.i.d. samples.


\section{Model}
\label{sec:model}

The model features an analyst and a principal.
The analyst (e.g., a researcher, an expert, or a firm) performs a statistical study whose content is summarized by the distribution over the parameters of a statistical model. 
We refer to this distribution as the analyst's belief over the model parameters.
The principal (e.g., a manager, a firm, or an agency) would like to elicit information on the analyst's belief. We follow the convention of using the female pronoun for the principal and the male pronoun for the analyst.

The analyst's belief captures his uncertainty about possible parameter values.
The statistical model is fixed exogenously and assumed well-specified.
Throughout the paper, $\Theta$ denotes the set of parameters of interest.
The process by which the analyst arrives at a belief is irrelevant and not part of the model; to fix ideas, the analyst may be a Bayesian statistician who starts from some prior over parameters and forms a posterior belief based on laboratory data.

The principal has access to an independent dataset. 
The data it includes are randomly generated and their distribution is a function of the parameter of the underlying statistical model. This relation is formalized by the concept of statistical experiments, following \citet{blackwell1951comparison} and \citet{le1996comparison}. A \df{statistical experiment} is a pair $(Y,\pi)$; $Y$ is the set of possible outcomes that captures the observables, and $\pi$ is a Markov kernel that associates, to every parameter value $\theta$, a distribution over outcomes, $\pi(\cdot | \theta)$.\footnote{By definition of a Markov kernel, $\theta \mapsto \pi(A | \theta)$ is measurable for every $A$.}
Thus, the outcome randomly generated by the experiment includes all the data observed by the principal.
Parameter and outcome spaces are assumed to be standard Borel spaces.\footnote{A standard Borel space is a separable completely metrizable topological space---i.e., a Polish space---endowed with its Borel $\sigma$-algebra.}
The goal beneath this assumption is to encompass a wide range of statistical models, from finite-dimensional parametric models to fully nonparametric ones.
The intuition behind our results is best described with finite spaces, which trivially satisfy this assumption.  

It is useful to recognize three classes of experiments. An experiment $(Y,\pi)$ is \df{categorical} if its outcome space, $Y$, is finite. It is \df{identified} when the underlying statistical model is identified, which means the parameter can be inferred from the outcome distribution: $\pi(\cdot|\theta) \ne \pi(\cdot|\theta')$ if $\theta \ne \theta'$. 
Intuitively, infinitely repeated draws from an identified experiment make it possible to infer perfectly the true parameter.
Finally, it is \df{complete} if, by analogy to statistical theory, the experiment is so that every probability distribution over outcomes can be induced by some belief over model parameters. Formally, for every distribution over outcomes $\mu \in \Delta(Y)$, there exists a belief $p \in \Delta(\Theta)$ such that for all measurable sets of outcomes $A \subseteq Y$,
\[
\mu(A) = \int_\Theta \pi(A|\theta) p(\de \theta).
\]

The principal offers incentives in the form of contingent payments, whereby the analyst is paid as a function of his report and the data observed by the principal after the analyst sent the report. Alternatively, we can think of the principal designing a performance score that the analyst wants to maximize on average.
Formally, for a given experiment $(Y,\pi)$, an elicitation mechanism, or simply \df{mechanism}, is a mapping $\phi:\scR \times Y \to \R$ that takes as input a report from a set $\scR$, together with an outcome randomly generated by the experiment, and outputs a payoff. The analyst cares about expected payoffs. In a direct mechanism the set of reports is the set of possible beliefs over parameters, $\scR = \Delta(\Theta)$. A direct mechanism is \df{incentive compatible} when
\begin{equation*}
	\E_p[\phi(p,y)] \ge \E_p[\phi(q,y)] \qquad\qquad \forall p,q \in \Delta(\Theta),
\end{equation*}
where $y$ is the outcome randomly generated by the experiment, and the notation $\E_p$ refers to the expectation operator in the case of $p$ being the parameter distribution.
There are, implicitly, two levels of randomness in the expectations above, first, a randomization over parameters according to $p$, and second, a randomization over outcomes according to $\pi$, given the draw of the parameter.
To simplify notation, when there is no ambiguity, the same symbol (e.g., $y$ or $\theta$) is used for a random element and a possible realization.

Incentive compatibility is a weak requirement trivially satisfied.
The literature typically restricts attention to mechanisms that satisfy strict incentive compatibility: Reporting one's true belief must be the only best response. In the present context, this criterion is not achievable in general, but we can still incentivize the truthful reporting of \emph{some} information.
We formalize what we learn from the analyst by means of information partitions \citep{aumann1976agreeing}. An \df{information partition} of distributions, $\scP$, or simply information for short, is a partition of the space $\Delta(\Theta)$. Two parameter distributions that belong to different  members of the partition are said to be \df{distinguishable under $\scP$}. They are indistinguishable otherwise. The concept of elicitability is then defined as follows:

\begin{definition}
An incentive-compatible mechanism $\phi$ \df{elicits information $\scP$} when, for any $p, q$ distinguishable under $\scP$, the strict inequality $\E_p[\phi(p,y)] > \E_p[\phi(q,y)]$ obtains.
\end{definition}
An information partition $\scP$ is \df{elicitable} with a given experiment when a mechanism for this experiment can be designed to elicit $\scP$.
Of course, a strict inequality alone cannot serve as a constraint on the power of incentives, but once established, it becomes possible to generate incentives of any desired magnitude by shifting and scaling the payoffs accordingly.


\section{Elicitable Information}
\label{sec:elicitable information}

The results of this section deal with elicitable information when the principal carries out a fixed arbitrary experiment.
We begin with the general case. 
In Sections~\ref{subsec:data made of multiple observations} and~\ref{subsec:data with covariates} we refine the experiment respectively by considering the case of multiple independent observations, and by accounting for observable covariates in the data.

Let us consider an arbitrary experiment $(Y,\pi)$ whose random outcome is $y$. 
Any belief over parameters, $p \in \Delta(\Theta)$, induces a distribution over outcomes, $\lambda_p \in \Delta(Y)$, defined by
\[
\lambda_p(A) = \int_\Theta \pi(A|\theta) p(\de \theta), 
\]
for each measurable set of outcomes $A$. The value $\lambda_p(A)$ is the mean probability $\E_p[\pi(A|\theta)]$ that $y \in A$. We refer to $\lambda_p$ as the \df{mean outcome distribution} associated to belief $p$.
We denote by $\scP^\star$ the information partition that captures the mean outcome distribution, meaning that two parameter distributions $p, q$ are indistinguishable under $\scP^\star$ if (and only if) they induce the same outcome distribution, i.e., $\lambda_p = \lambda_q$.

The result below characterizes the information an arbitrary experiment enables us to elicit.
\begin{theorem}
	\label{thm:baseline}
	An information partition $\scP$ is elicitable with $(Y,\pi)$ if, and only if, $\scP$ is coarser than $\scP^\star$.
\end{theorem}
By coarser and finer, we always mean weakly coarser and weakly finer, respectively. In the sequel, we shall refer to $\scP^\star$ as the \df{maximal information elicitable with $(Y,\pi)$}.

The proof, in Appendix~\ref{appx:Proofs of Elicitable information}, is simple. When payment contingencies are on experiment outcomes, the analyst only cares to assess the outcome distribution, all other information is superfluous. Since this distribution is itself determined by the model parameter, which the analyst is typically unsure about, the analyst's belief over outcomes becomes the mean outcome distribution that results from these two layers of randomization.
Eliciting this information is well-known in special cases of sets $Y$ such as finite and finite-dimensional sets. While not possible for arbitrary sets $Y$, eliciting this information in the more general cases allowed under our assumptions on $Y$ remains possible by a mechanism whose functioning is equivalent to randomizing over a large enough collection of quadratic scoring rules.

One implication of this theorem is that if, for each index $i$ of an arbitrarily rich set, some information partition $\scP_i$ is elicitable with experiment $(Y,\pi)$, possibly using a specific mechanism for each value of $i$, then the join information is elicitable as well. In other words, the information partitions $\scP_i$ are elicitable all at once with a unique mechanism. This observation is instrumental to many examples of our paper.\footnote{This observation owes to the structure we impose on the outcome space. For example, consider an arbitrary space of states of the world and a family of events $E_i$ with $i \in [0,1]$ whose possible occurrence is eventually observed. We can of course elicit the probability of each $E_i$ separately (e.g., with a quadratic scoring rule) but it is not possible to elicit the probabilities of all $E_i$'s at once with the same mechanism in general. For a detailed discussion of the issues involved, see \citet[p. 398]{chambers2021dynamic}.}

An information of common interest is point estimates. 
The two corollaries below give necessary and sufficient conditions for the elicitation of the mean estimate of some real function of the parameter.
In these results, and throughout the paper, $\scG$ refers to the set of bounded measurable functions from $\Theta$ to $\R$. The focus on bounded functions is identical to asking that the expected value of $g(\theta)$ is finite under all beliefs about $\theta$, and so ensures that the task of elicitation is well defined. 
Saying the mean of $g(\theta)$ is elicitable with $(Y,\pi)$ is formally saying that the information partition that captures the mean of $g(\theta)$ is coarser than $\scP^\star$, i.e., that $\E_p[g(\theta)] = \E_q[g(\theta)]$ whenever $p, q$ are indistinguishable under $\scP^\star$.

\begin{corollary}
	\label{cor:estimator}
	If $g \in \scG$ and there exists $w:Y\to\R$ with
	$g(\theta) = \int_Y w(y) \pi(\de y | \theta)$,
	then the mean of $g(\theta)$ is elicitable with $(Y,\pi)$.\footnote{We implicitly assume that $w$ is integrable with respect to $\pi(\cdot | \theta)$ for every $\theta$.}
\end{corollary}

This corollary is a direct implication of the law of iterated expectations, sometimes referred to as the martingale property. A short formal proof is in Appendix~\ref{appx:Proofs of Elicitable information}. The converse does not hold in general, but does hold when the experiment is categorical.

\begin{corollary}
	\label{cor:estimator-converse}
	If $(Y,\pi)$ is categorical, $g \in \scG$, and the mean of $g(\theta)$ is elicitable with $(Y,\pi)$, then there exists $w:Y\to\R$ with $g(\theta) = \sum_y w(y) \pi(y | \theta)$.
\end{corollary}
Hence, in the categorical case, the mean of $g(\theta)$ is elicitable precisely when there exists an unbiased estimator for $g(\theta)$ taking as input the principal's observations.
This second corollary, proved in Appendix~\ref{appx:Proofs of Elicitable information}, is an application of the Fundamental Theorem of Duality.
The examples below illustrate these results.

\begin{example}[The German tank problem]
\label{ex:german tank problem}
Consider a population that includes an unknown number of units numbered consecutively starting from 1, together with the experiment that consists in observing just one unit drawn at random from this population (each unit has the same probability of a draw).\footnote{The German tank problem is a classical example in statistics about the estimation of the maximum of a discrete uniform distribution, whose original application is the estimation of the production of tanks and other military items manufactured during World War II \citep{goodman1954some}.}
In this example, the underlying statistical model is the uniform discrete distribution. The parameter $\theta$ is the population size, the outcome is the numbered unit, thus $\Theta=Y=\{1,2,\dots\}$, and the Markov kernel is
\[
\pi(k|\theta) = 
\begin{cases} 
1/\theta &\text{if } k \le \theta, \\ 
0 &\text{if } k > \theta. 
\end{cases}
\]
With this single data point, it is possible to elicit the full information on the analyst's belief about $\theta$.
To see this, we apply Corollary~\ref{cor:estimator} with the functions $g_m:\Theta\to\R$ defined as
\[
g_m(\theta) = \sum_{k=1}^\infty w_m(k) \pi(k | \theta)
\]
for $m \in \N_+$ and
\[
w_m(k) = 
\begin{cases} 
1 &\text{if } k \le m, \\ 
-m &\text{if } k = m+1,\\ 
0 &\text{if } k > m+1. 
\end{cases}
\]
Corollary~\ref{cor:estimator} says that the mean of every $g_m(\theta)$ is elicitable, and by Theorem~\ref{thm:baseline} these means are elicitable all at once with a same mechanism. Observe that $\E[g_m(\theta)] = \Pr[\theta \le m]$, since $g_m(\theta)=0$ if $\theta \ge m+1$ and $g_m(\theta)=1$ if $\theta \le m$. Therefore, the c.d.f. of $\theta$, which captures the full information on the analyst's belief, is elicitable. 
\end{example}

\begin{example}[Poisson model]
\label{ex:poisson model}
Consider a population of individuals who borrow money from a bank or an investor. There is a risk of default, and the number of defaults over a reference time period is modeled as a Poisson distribution, whose parameter $\theta$ captures the default rate, presumed constant for the term of interest. 
In this context, the experiment that reveals the number of loan defaults over the reference period has the outcome set $Y=\{0,1,2,\dots\}$ and the Markov kernel
\[
\pi(k|\theta) = \frac{\theta^k}{k!} e^{-\theta}.
\]
With such data, it is also possible to elicit the full information on the analyst's belief.
To see this, let us define for each $t \in \R$ the function on outcomes $w_t(k) = (1+\I t)^k$ where $\I$ is the imaginary unit,\footnote{Alternatively, we can leverage results on the identification of mixtures of distributions. The fact that mixtures of Poisson distributions are identified in the statistical sense (see p. 245 of \citealp{teicher1961identifiability}) implies that one can infer exactly the distribution over Poisson parameters from the induced distribution over outcomes, which the mechanism used in the proof of Theorem~\ref{thm:baseline} elicits indirectly. We will make use of this fact in our example on Gaussian linear models.}
and then let
\begin{equation*}
	g_t(\theta) 
		= \sum_{k=0}^\infty w_t(k) \pi(k|\theta) 
		= e^{-\theta} \sum_{k=0}^\infty \frac{(\theta+\I t \theta)^k}{k!} 
		= e^{\I t \theta}.
\end{equation*}
We apply Corollary~\ref{cor:estimator} and Theorem~\ref{thm:baseline}, and get that the means of all $g_t(\theta)$ are elicitable at once. Observing that $t \mapsto \E[e^{\I t \theta}]$ is the characteristic function of $\theta$, the above experiment elicits the full distribution over $\theta$---and so the full information on the analyst's belief.
\end{example}

\begin{example}[Bernoulli model]
\label{ex:bernoulli model}
Let us return to the example raised in the Introduction.
There is a large population of individuals---a continuum---and a new drug is effective on an unknown fraction of the population. This fraction is the parameter of interest, and the underlying statistical model is the Bernoulli model. 
We look at the experiment that consists in performing a single clinical trial on a random individual from the population. 

Here, the set of possible parameters is $\Theta=[0,1]$, and the experiment is described by the binary outcome set $Y=\{0,1\}$ (outcome 1 if the drug is effective on the individual and outcome 0 otherwise) with the Markov kernel
\begin{align*}
\pi(1|\theta) &= \theta,\\
\pi(0|\theta) &= 1- \theta.
\end{align*}
By Theorem~\ref{thm:baseline}, the elicitable information is precisely the mean fraction of the population on which the vaccine is effective. No further information can be elicited with a single observation.

\end{example}


\subsection{Data Made of Multiple Observations}
\label{subsec:data made of multiple observations}

Suppose the data the principal collects is composed of independent observations from one or more experiments. 
Specifically, consider $n$ arbitrary experiments $(Y_i,\pi_i)$, $i=1,\dots,n$, which may be identical, and let $y_i$ denote the outcome randomly and independently generated by $(Y_i,\pi_i)$, conditionally on the model parameter. The principal observes the vector $y=(y_1, \dots, y_n) \in Y_1 \times \cdots \times Y_n$. The experiment that generates $y$ is a compound experiment, which is the product of the $n$ elementary experiments, written $(Y_1,\pi_1) \otimes \cdots \otimes (Y_n,\pi_n)$. 

Of course, Theorem~\ref{thm:baseline} continues to apply to this compound experiment, but it can be more convenient to work with the individual experiments that make the product. Our second result, below, provides sufficient conditions for mean information to be elicitable overall given knowledge of information elicitable with these individual experiments.

\begin{theorem}
	\label{thm:multiple}
	If, for $i=1, \dots, n$, the experiment $(Y_i,\pi_i)$ elicits the mean of $g_i(\theta)$ with $g_i \in \scG$, then the product experiment elicits the mean of $g(\theta)$ with $g = g_1 \times \cdots \times g_n$.
\end{theorem}
The proof of Theorem~\ref{thm:multiple} is in Appendix~\ref{appx:Proofs of Elicitable information}.

Several implications are worth noting. First, recall that by Theorem~\ref{thm:baseline}, with one observation from an experiment, we can elicit the mean probability of each set of outcomes. Combining this fact with Theorem~\ref{thm:multiple}, we conclude that collecting multiple independent observations from the same experiment enables us to elicit the higher order moments. 
Second, for any experiment and any $g \in \scG$, if the mean of $g(\theta)$ is elicitable with $n$ observations from this experiment, then the variance of $g(\theta)$ is elicitable with $2 n$ observations: Since the variance is written $\E[g(\theta)^2] - \E[g(\theta)]^2$, this fact follows by applying Theorem~\ref{thm:multiple} on the product of the two experiments that each supply $n$ observations. Similarly, if we have two experiments $(Y,\pi_Y)$ and $(Z,\pi_Z)$ with which we can elicit the mean of $g_Y(\theta)$ and $g_Z(\theta)$, respectively, then their covariance is elicitable with one observation from each experiment.

The examples below illustrate simple applications of Theorem~\ref{thm:multiple}. In Section~\ref{subsec:data with covariates}, we put the theorem to work for several statistical models.

\begin{example}[Elicitation of the expertise level]
Suppose the principal collects two independent observations from an arbitrary experiment $(Y,\pi)$ that is identified. 
The composite experiment that produces the principal's data is the product $(Y,\pi) \otimes (Y,\pi)$.

\emph{We claim that this experiment enables the principal to learn whether the analyst knows the parameter and, when the analyst does know, to elicit the true parameter value.}
By knowing the parameter, we mean that the analyst is absolutely certain about the parameter value.
Put formally, this means that under the information partition that captures the maximal information elicitable by this experiment, the probability measure that puts full mass on a single parameter value is distinguishable from any other belief in $\Delta(\Theta)$. The argument goes as follows.

In the proof of Theorem~\ref{thm:baseline}, we have shown the existence of a countable collection of measurable subsets of $Y$, $\{E_i : i \in \N\}$, such that for any two probability distributions over $Y$, $\mu$ and $\nu$, we have $\mu=\nu$ if and only if for all $i$, $\mu(E_i)=\nu(E_i)$; that is, outcome distributions are identified on this countable collection of sets. 
And since the experiment is identified, any parameter value $\theta_0$ is uniquely identified by the values $\pi(E_i|\theta_0)$, $i \in \N$.

Let $g_i(\theta) = \pi(E_i|\theta)$.
By Theorem~\ref{thm:baseline}, experiment $(Y,\pi)$ elicits the means of all $g_i(\theta)$. Thus one observation enables the principal to learn the mean of every $g_i(\theta)$, and by Theorem~\ref{thm:multiple}, the two observations makes it possible to elicit the variance of every $g_i(\theta)$.

Let $p \in \Delta(\Theta)$ be the analyst's belief. 
To assess if the analyst knows the true parameter, the principal can examine  the variances elicited. 
If $\var_p[g_i(\theta)]=0$ for all $i$, then, letting $\gamma_i = \E_p[g_i(\theta)]$, we have that, for every $i$,  with $p$-probability 1, $g_i(\theta) = \gamma_i$. Since there are countably many indices $i$, we may reverse the order of the quantifiers and get that with $p$-probability 1, $g_i(\theta) = \gamma_i$ for all $i$. Hence, there exists $\theta_0$ such that $p(\theta=\theta_0)=1$: The analyst believes he knows the parameter. And conversely, if the analyst believes the parameter value is $\theta_0$ for sure, then clearly $\var_p[g_i(\theta)]=0$.

Next, if, after examining the variances, the principal concludes that the analyst knows the true parameter parameter value, here referred to as $\theta_0$, the principal infers $\pi(E_i|\theta_0)$ from the values of the elicited means $\E_p[g_i(\theta)]$. The property of identification mentioned above guarantees that $\theta_0$ may be deduced from the values $\pi(E_i|\theta_0)$, $i \in \N$.
\end{example}

\begin{example}[Elicitation of probability density functions]
When the analyst is not sure about the parameter but may still possess valuable information, the principal can work with more than two observation to approximate the analyst's estimated density function over parameters.

Consider a categorical experiment $(Y,\pi)$ that is identified, and let us enumerate the elements of $Y$ from 1 to $m$.
We may, without loss of generality, ask that the parameter $\theta=(\theta_1, \dots, \theta_m)$ represents the outcome distribution and take for $\Theta$ the $(m-1)$-simplex of $\R^m$ (the corresponding model is sometimes referred to as the multinoulli model).

Suppose the researcher's belief on $\theta$ is captured by a Lipschitz-continuous probability density function $f$. We will demonstrate the following claim. \emph{With $n$ independent observations, the principal can elicit from the analyst a function $\hat f$ that approximates the true density of the analyst according to}
\begin{equation*}
	\int_\Theta |f(\theta) - \hat f(\theta)|^2 \de \theta = O(1/n).
\end{equation*}
Notice that the left hand side is the mean integrated squared error commonly used in statistics for the estimation of probability densities.

The proof of this claim builds once again on Theorem~\ref{thm:multiple}.
We begin by defining two classes of functions, $\scC_D$ will be the class of all Lipschitz-continuous densities over $\Delta^{m-1}$, and $\scC_P$ will be the class of all polynomials of $m-1$ variables of degree $n$.
Let $P_1, \dots, P_K$ be an orthonormal basis of $\scC_P$ with respect to the inner product on $L^2$; that is, a basis that satisfies
\begin{equation*}
\int_\Theta P_i P_j = \1\pc{i=j},
\end{equation*}
where we use the shorthand integral notation.
Such bases are easily constructed using the Gram–Schmidt process of orthonormalization, or even more simply using the tensor product basis from Legendre polynomials.

From Theorem~\ref{thm:multiple}, we know that with $n$ independent observations, we can elicit the mean, $c_k$, of every $P_k(\theta_1, \dots, \theta_m)$ for $k=1, \dots, K$. With this information, we can compute the polynomial 
\[
P = \sum_k c_k P_k
\quad\text{ with }\quad
c_k = \int_\Theta P_k f.
\]
We observe that this polynomial is the orthogonal projection of $f$ on $\scC_D$, hence,
\begin{equation*}
\int_\Theta (f - P)^2 = \min_{Q \in \scC_P} \int_\Theta (f - Q)^2.
\end{equation*}
By the multidimensional version of Jackson's inequality (Theorem 2 of \citealp{ganzburg1981multidimensional}), there exists a constant $\gamma$ that depends on $f$ such that
\begin{equation*}
\min_{Q \in \scC_P} \int_\Theta (f - Q)^2 \le \gamma/n.
\end{equation*}

We conclude that, from the information that can be elicited truthfully from the researcher, we can infer an approximation $\hat f$ of the density $f$ that satisfies
\begin{equation*}
\int_\Theta |f - \hat f|^2 = O(1/n)
\end{equation*}
where $n$ is the size of the principal's dataset.
\end{example}


\subsection{Data with Covariates}
\label{subsec:data with covariates}

Suppose the data observed by the principal includes features or characteristics of the units of observation. 
The principal now observes a pair $(x,y)$ where $y \in Y$ continues to be called the outcome and $x \in X$ is a covariate, in a general sense, the set of possible covariate values being arbitrary. The distribution of covariates is known and part of the model. Both $X$ and $Y$ are standard Borel spaces.

The experiment that generates the random pair (covariate,outcome) can be represented as a compound experiment made of two parts. 
First, a covariate $x$ is randomly generated. 
Second, an outcome $y$ is randomly generated by an experiment $(Y,\pi_x)$, where the subscript $x$ captures the possible dependence of the outcome on the covariates. 
These compound experiments are mixtures of elementary experiments. Formally, a mixture of elementary experiments $(Y,\pi_x)$, $x \in X$, is an experiment $(X \times Y, \pi)$ that satisfies
\[
	\pi(A \times B) = \int_A \pi_x(B | \theta) \mu(\de x),
\]
where $x \mapsto \pi_x(B | \theta)$ is measurable, and where $\mu$ is the probability distribution over covariates. 

Of course, Theorem~\ref{thm:baseline} continues to apply, but it can be more convenient to work with the elementary experiments that make the mixture. Our second result links the information elicited by elementary experiments to information that can be elicited by the mixture. It provides sufficient conditions on the information elicited when the principal has access to covariate data. 

In the theorem below, we consider any mixture experiment of the form described above, and let $\scP(x)$ be an information partition elicitable with $(Y,\pi_x)$.
\begin{theorem}
	\label{thm:covariates}
	Let $\scP_0$ be any information partition and $\{ x_i\}$ be a finite or countably infinite collection of covariates independently drawn from $\mu$. 
	If, with nonzero probability, the join of the partitions in $\{ \scP(x_i) \}$ is finer than $\scP_0$, then $\scP_0$ is elicitable with the mixture experiment.
\end{theorem}
The proof of Theorem~\ref{thm:covariates} is in Appendix~\ref{appx:Proofs of Elicitable information}. The examples that follow illustrate the use of both Theorems~\ref{thm:multiple} and~\ref{thm:covariates} in the context of various statistical models.

\begin{example}[Gaussian linear model]
\label{ex:gaussian linear model}
This example considers the Gaussian linear model that relates a real-valued dependent variable $y$ (the outcome) to a covariate vector $x=(x^{(1)},\dots,x^{(K)})\in\R^K$ by
\begin{equation*}
	y = \beta_0 + \beta_1 x^{(1)} + \cdots + \beta_K x^{(K)} + \sigma \epsilon,
\end{equation*}
where $\epsilon$ is a standard normal error. The model parameter is the $(K+2)$-dimensional vector composed of the vector of coefficients $\beta=(\beta_0, \dots,\beta_K) \in \R^{K+1}$ and the variance of the error term, $\sigma^2 \in \R_+$.
Let $\theta = (\beta,\sigma^2) \in \Theta$ and let $\mu$ be the probability distribution over covariates. Assume the support of $\mu$ has a nonempty interior and $\Theta = I^{K+2}$ where $I$ is a compact interval of $\R$. 

\emph{We claim that under the Gaussian linear model just described, the full information of the analyst is elicitable with just one observation.} 

The proof of this claim is composed of two simple steps, noting that the relevant experiment is a mixture experiment that generates the data point $(x,y)$. 
First, fix an arbitrary $x \in X$ and write $h(x) = \beta_0 + \beta_1 x^{(1)} + \cdots + \beta_K x^{(K)}$.
A belief (probability distribution) over $\theta$ induces a belief over $(h(x),\sigma^2)$. Mixtures of Gaussian distributions are identifiable under the condition that the mean and variance of the individual Gaussians are in compact set \citep{bruni1985identifiability}. 
Therefore, the distribution over $y$ allows the inference of the analyst's belief over $(h(x),\sigma^2)$, and, for every $a,b \in \R$, of the distribution of $a h(x) + b \sigma^2$---that is, we learn the distribution of every one-dimensional projection of $(h(x),\sigma^2)$. 
Let $\scP(x)$ be the information partition that captures the knowledge of each of these distributions. 
By Theorem~\ref{thm:baseline}, this information partition is elicited from the elementary experiment that generates $y$ conditionally on $x$ following the Gaussian model just described.

Second, let $\scP_0$ be the information partition associated with the full information on the analyst's belief, $\scP_0 = \{ \{p\} : p \in \Delta(\Theta)\}$, and let $x_1, x_2, \dots,$ be an infinite sequence of covariates independently drawn according to $\mu$. 
By the assumption imposed on $\mu$, with positive probability, this sequence is dense in some open set $O \subset \R^K$. Any finite-dimensional distribution is uniquely determined by its one-dimensional projections \citep{cramer1936theorems} and hence, the density of the sequence of covariates in $O$ together with the fact that a multivariate c.d.f. is right-continuous implies that with positive probability, the join of the partitions in $\{\scP(x_i)\}$ is finer than $\scP_0$. The claim then follows from a direct application of Theorem~\ref{thm:covariates}. 
\end{example}

\begin{example}[Semiparametric linear model]
\label{ex:semiparametric linear model}
This example extends the model of Example~\ref{ex:gaussian linear model}. We continue to postulate the linear relationship
\[
	y = \beta_0 + \beta_1 x^{(1)} + \cdots + \beta_K x^{(K)} + \sigma \epsilon
\]
but we no longer restrict the shape of the error term $\epsilon$. The model parameter $\theta$ is now composed of the finite-dimensional vector $(\beta_1, \dots, \beta_K) \in I^K$ (where $I$ continues to be a compact interval of $\R$) and a real distribution $\theta_\epsilon$ with $\theta_\epsilon = \beta_0 + \sigma \epsilon$ such that $\beta_0, \sigma^2 \in I$. 
It is convenient to work with the parameter space $\Theta = I^K \times \Theta_\epsilon$, where the parametric part of the model, $I^K$, is the space of possible values for the coefficients of the independent variables, and the nonparametric part of the model, $\Theta_\epsilon$, is the space of all real distributions whose mean and variance belong to $I$. This space is rich but satisfies the assumption of Section~\ref{sec:model}.\footnote{Specifically, the parameter space is a Polish space as a product of the Euclidean space $I^K$ and the space $\Theta_\epsilon$ endowed with the Wasserstein distance between distributions.}
Let $\mu$ be the covariate distribution and assume that for every $\alpha \in\R^K$, with nonzero probability, the covariate $x$ satisfies $\alpha_1 x^{(1)} + \cdots + \alpha_K x^{(K)} \not\in \{0,1\}$. This assumption is weaker than the assumption on $\mu$ made in the preceding example. 

\emph{We claim that under the semiparametric linear model just described, and with one observation, the principal can elicit from the analyst the mean belief of $\beta_0, \dots,\beta_K$.}

To prove the claim, we first take any $x \in X$ and let $h(x) = \beta_0 + \beta_1 x^{(1)} + \cdots + \beta_K x^{(K)}$. Corollary~\ref{cor:estimator} shows that the mean of $h(x)$ is elicitable with the experiment that generates $y$ randomly conditionally on $x$ (this is obtained by using $w(y) = y$).
Let $\scP(x)$ be the corresponding information partition.

Second, let $\scP_0$ be the information partition associated to the mean belief of $\beta_0, \dots,\beta_K$, and let $x_1, \dots, x_{K+1}$ be drawn independently at random from $\mu$. For any given belief $p$ over $\theta$, the means of $\beta_0, \dots, \beta_K$ satisfy the equalities
\[
\E_p[\beta_0] + \E_p[\beta_1] x_i^{(1)} + \cdots + \E_p[\beta_K] x_i^{(K)} = \E_p[h(x_i) | x_i] \qquad \forall i =1, \dots, K+1.
\]  
The assumption imposed on $\mu$ implies that with positive probability, the matrix
\[
\Matrix{1,x_1^{(1)},\dots,x_1^{(K)};\vdots,\vdots,\ddots,\vdots;1,x_{K+1}^{(1)},\dots,x_{K+1}^{(K)}}
\]
has full rank, which, in turn, implies that with positive probability, one can infer the mean of each $\beta_i$ from knowledge of $\E_p[h(x_i) | x_i]$. And hence, with positive probability, the join of $\scP(x_1), \dots, \scP(x_{K+1})$ is finer $\scP_0$. We can then apply Theorem~\ref{thm:covariates} to get the claim. 
\end{example}

\begin{example}[Nonparametric classification]
\label{ex:nonparametric classification}

We now look at a nonparametric model of classification in which labels take binary values 0 or 1. 
The set of possible vectors of features, $X$, is a compact subset of $\R^K$ for $K \ge 1$.
In this model, the parameter $\theta$ is the likelihood function that maps feature vectors to the range $[0,1]$, so  $\theta(x)$ is the probability of observing label 1 for a unit with feature vector $x$. 
We assume $\theta \in \Theta$, with $\Theta$ the set of all continuous functions from $X$ to $[0,1]$. 
This space satisfies the assumption in Section~\ref{sec:model}.\footnote{The space of continuous functions from a compact metric space to a Polish space is Polish by Theorem 4.19 of \citet{kechris1995classical}.}

Suppose the distribution over covariates, $\mu$, has full support, and consider once again the experiment that generates just one random observation $(x,y)$, where $x$ is the feature and $y$ is the label of a unit randomly drawn.
We claim that \emph{under the nonparametric classification model just described, with one observation, the principal can elicit from the analyst the mean label likelihood function.}

First, observe that for an arbitrary dense subset $D$ of $X$, if for every $x \in D$ the mean of $\theta(x)$ is identical when $\theta$ is distributed according to two different beliefs over parameters, $p$ and $q$, then by continuity of the mapping $x \mapsto \int_\Theta \theta(x) p(\de \theta)$, the mean of $\theta(x)$ continues to be identical for every $x \in X$. That the mapping is continuous follows from the Dominated Convergence Theorem. Second, Example~\ref{ex:bernoulli model} shows that the experiment that generates, at random, a label associated with a given feature vector $x$ elicits the mean of $\theta(x)$. Let $\scP(x)$ be the information partition that captures this information.

Finally, let $\scP_0$ be the information partition associated with the mean label likelihood function and $x_1, x_2, \dots,$ be an infinite sequence of covariates independently drawn according to $\mu$. 
With probability 1, this sequence is dense in $X$, so by the observation above, with probability 1, the join of the partitions $\scP(x_1), \scP(x_2), \dots,$ is a refinement of $\scP_0$. We conclude with an application of Theorem~\ref{thm:covariates}.
\end{example}


\subsection{Discussion}
\label{subsec:applications}

We conclude this section with a discussion on various aspects of our framework.

\subsubsection{On Data Requirements for Complete Elicitation}

We have argued that the full information on the analyst's belief is generally not elicitable, and we also have given examples when it is.
The ability to elicit the full information depends on the relative size of the space of the analyst's beliefs and the space of the possible outcome distributions.

Suppose the principal collects independent observations from a categorical, identified experiment. We make two claims.

\emph{The first claim is that if the parameter space is infinite and the principal gets finitely many data points, it is never possible to elicit the complete information of the analyst (his full belief).} To elicit the full information, we need to have an infinite set of observations, and this also makes it possible to infer the true value of the parameter by statistical inference.
In terms of data requirements, complete elicitation is as demanding as perfect parameter estimation.

\emph{The second claim is that if, on the contrary, the parameter space is finite of size $n$, then the principal never needs more than $n-1$ observations to elicit the analyst's full information.}

We begin by proving the first claim.

The product experiment that generates the principal's data is categorical. Let $m$ be the size of its outcome set, which we enumerate implicitly.
Suppose the parameter space is finite of size $n$ with $n > m$, so that $\Theta = \{\theta_1, \dots, \theta_n\}$.
Of course, the argument carries over to all larger parameter spaces.
Consider the following linear operator from $\R^n$ to $\R^m$:
\begin{equation*}
L(p_1, \dots, p_n) = \sum_i p_i \pi(\cdot | \theta_i),
\end{equation*}
where a distribution over outcomes takes values in the $(m-1)$-simplex and so is identified with a vector of $\R^m$.
This operator transforms a parameter distribution into the mean outcome distribution, for the composite outcome generated by the product experiment.

The maximal information that the product experiment elicits is precisely the mean outcome distribution (Theorem~\ref{thm:baseline}). Thus, if we could elicit the full parameter distribution, then any pair of distinct beliefs over parameters would induce two different mean outcome distributions, which is impossible because the dimension of the domain of $L$ exceeds the dimension of the codomain.

To demonstrate the second claim, we note that since the experiment is both identified and categorical, there exists a one-to-one function $g \in \scG$ such that the mean of $g(\theta)$ is elicitable.
Let $\Theta = \{\theta_1, \dots, \theta_n\}$. 
By Theorem~\ref{thm:multiple}, with $n-1$ data points, we can elicit the mean of $g(\theta)^k$ for $k=1, \dots, n-1$.
Consider the following linear operator from $\R^n$ to $\R^n$:
\begin{equation*}
L(p_1, \dots, p_n) = \sum_i p_i \Matrix{1;g^1(\theta_i);\vdots;g^{n-1}(\theta_i)}.
\end{equation*}
This operator transforms a belief over parameters to the mean distribution of the vector $(1,g^1(\theta),\dots,g^{n-1}(\theta))$. The transformation is one-to-one (it is represented by a full rank Vandermonde matrix), and thus the analyst's belief over parameters is fully determined by the mean of each $g(\theta)^k$, so that the full belief is elicitable.

\subsubsection{On the Elicitation of Point Estimates}

Mean point estimates appear to play a special role in our framework.
Our first result states that one can always elicit the mean of the outcome distribution, and we have argued by example that the mean belief on some parameters can, sometimes, be elicited with just a few data points. It turns out the other two standard point estimates that are the mode and the median are difficult to elicit in that they demand more data for incentive provision---in fact, they demand as much data as for the elicitation of the full information.

At a general level, this fact is a consequence of the geometry of the maximal information elicited, $\scP^\star$. This partition takes the form of parallel, or translated, linear spaces. The information partition of the mean is also composed of parallel linear spaces, but not that of the mode or median. We illustrate the case of the mode below, the case of the median is similar. To simplify matters, we focus on a finite, real parameter space: $\Theta = \{\theta_1, \dots, \theta_n\} \subset \R$ with $\theta_1 < \cdots < \theta_n$. 

\emph{Take an arbitrary experiment. We claim that if the mode of $\theta$ is elicitable with this experiment, then the full parameter distribution is also elicitable.} 

By contradiction, suppose the experiment does not elicit the full parameter distribution but that it elicits the mode.
Suppose $p$ and $p'$ are two beliefs over $\theta$ are indistinguishable under $\scP^\star$:
\[
\sum_i p(\theta_i) \pi(\cdot | \theta_i)
=
\sum_i p'(\theta_i) \pi(\cdot | \theta_i).
\]
Let $v = p'- p$. Note that $\sum_i v(\theta_i)=0$. Then note that for all beliefs $q$ and $q'$ such that $q' = q + \alpha v$ for some $\alpha$, we also have
\[
\sum_i q(\theta_i) \pi(\cdot | \theta_i)
=
\sum_i q'(\theta_i) \pi(\cdot | \theta_i),
\]
so that $q$ and $q'$ are also indistinguishable under $\scP^\star$.

Thus, if the experiment elicits a mode, then $q$ and $q'$ must share a mode in common. Let $q$ be the uniform distribution, and let $I = \argmax_i v(\theta_i)$ and $J = \argmin_j v(\theta_j)$. Note that $I \cap J = \emptyset$.
If $\alpha$ is small enough, then $q +\alpha v \in \Delta(\Theta)$. In addition, if $\alpha > 0$ then $q + \alpha v$ has all modes in I, and if $\alpha < 0$ then $q + \alpha v$ has all modes in J. Hence a contradiction.


\section{Comparison of Experiments}
\label{sec:comparison of experiments}

So far we have focused on one fixed experiment.
When data take multiple forms, can be collected in different ways, or simply when various amounts of data can be gathered, the principal has a choice to make about which experiment to carry out.

A standard way to compare experiments is Blackwell's informativeness order.
An experiment $(Y,\pi_Y)$---that generates random outcome $y$---is more informative than an experiment $(Z,\pi_Z)$---that generates random outcome $z$---if we can write $z=h(y,\epsilon)$ for some function $h$ and some independent random noise $\epsilon$ (and where equality is in distribution); that is, the second experiment is a garbling of the first. Intuitively, we can emulate the data supplied by $(Z,\pi_Z)$ by transforming the data supplied by $(Y,\pi_Y)$, possibly with the addition of random noise.
Blackwell's comparison is relevant when data is used for statistical inference, because the statistician who cares to minimize expected losses will be better off with an experiment that is more informative in the sense of Blackwell.

Rather than estimate model parameters, we utilize experiments as incentive generator to elicit information on these parameters. This goal motivates two alternatives to compare experiment. Our first basis of comparison, explored in Section~\ref{subsec:Comparison based on Elicitability}, is the information elicitable with a given experiment.
Our second basis of comparison, detailed in Section~\ref{subsec:Comparison based on Incentives}, is the power of incentives that can be implemented. We will see that these two approaches are essentially equivalent. 

Throughout this section, the focus is on experiments that are categorical. 
By `experiment' we always mean categorical experiment.
We also assume that the outcomes associated with these experiment are enumerated, and leave the enumeration implicit. For any two (enumerated) sets of outcomes $Y$ and $Z$, a $Y$-by-$Z$ matrix $M$ stands for a matrix with $|Y|$ rows and $|Z|$ columns, $M(y,z)$ is the entry in row $y$ and column $z$. When $M$ is a Markov matrix---a matrix with nonnegative entries whose rows sum to 1---we also use the familiar notation $M(z | y)$, interpreted as the transition probability from $y$ to $z$.


\subsection{Comparison based on Elicitability}
\label{subsec:Comparison based on Elicitability}

Let us say that \df{an experiment dominates another experiment in the sense of elicitation} if every information partition elicitable with the latter is also elicitable with the former. 
This domination relation defines an `elicitation' order on experiments, just like Blackwell's informativeness order.\footnote{Strictly speaking, these are quasi-orders.} 
Although Blackwell's informativeness order and the elicitation order serve very different purposes, they are closely related.

To see this, consider two experiments $(Y,\pi_Y)$ and $(Z,\pi_Z)$. 
Recall that the Blackwell order can be restated as follows: $(Y,\pi_Y)$ is more informative than $(Z,\pi_Z)$ if, and only if, there exists a $Y$-by-$Z$ Markov matrix $M$ such that
\begin{equation}
\label{eq:dominance-Blackwell}
\pi_Z(z |\theta) = \sum_{y \in Y} M(z|y) \pi_Y(y|\theta) \qquad  \forall z \in Z, \theta \in \Theta.
\end{equation}

On the other hand, $(Y,\pi_Y)$ dominates $(Z,\pi_Z)$ in the sense of elicitation if, and only if, 
there exists a $Y$-by-$Z$ matrix $M$ (which need not be Markovian) such that
\begin{equation}
\label{eq:dominance-elicitation}
\pi_Z(z |\theta) = \sum_{y \in Y} M(y,z) \pi_Y(y|\theta) \qquad  \forall z \in Z, \theta \in \Theta.
\end{equation}
Indeed, the mean distribution of the outcomes generated by $(Z,\pi_Z)$ is elicitable with $(Z,\pi_Z)$ (Theorem~\ref{thm:baseline}), and so it is also elicitable with $(Y,\pi_Y)$ if $(Y,\pi_Y)$ dominates $(Z,\pi_Z)$.
Corollary~\ref{cor:estimator-converse} then yields the existence of a $Y$-by-$Z$ matrix $M$ that satisfies Equation~\eqref{eq:dominance-elicitation}. Of course, if Equation~\eqref{eq:dominance-elicitation} holds for some matrix $M$, then $(Y,\pi_Y)$ dominates $(Z,\pi_Z)$ in the sense of elicitation. 

In the sequel we abuse notation slightly and summarize Equations~\eqref{eq:dominance-Blackwell} and~\eqref{eq:dominance-elicitation} by $\pi_Z = \pi_Y M$, even though $\pi_Z$ and $\pi_Y$ are not matrices when the parameter space is infinite.
The next lemma summarizes the discussion above.
\begin{lemma}
The experiment $(Y,\pi_Y)$ dominates the experiment $(Z,\pi_Z)$ in the sense of elicitation if, and only if, $\pi_Z = \pi_Y M$ for some $Y$-by-$Z$ matrix $M$.
\end{lemma}
Remarkably, the only difference between Blackwell dominance and dominance by elicitation is that $M$ is a Markov matrix in the former, while it is a general matrix in the latter. As an immediate consequence, any experiment that dominates another in the sense of Blackwell also dominates in the sense of elicitation, but not the opposite: The elicitation order is ``more complete'' than the informativeness order.

To illustrate the difference, let us go back to the Bernoulli model evoked in Example~\ref{ex:bernoulli model}.
Let $(\{0,1\},\pi)$ be the experiment that corresponds to a single clinical trial, as in the original example:  $\pi(1 | \theta) = \theta$, $\pi(0|\theta)=1-\theta$.
In addition, let $(\{0,1\},\pi')$ be the experiment that adds noise to the clinical trial: $\pi'(1 | \theta) = .05 + .9 \theta$, $\pi'(0 | \theta) = .95 - .9 \theta$. In words, with 90\% chance, the outcome the experiment generates is the true result of the clinical trial, and with 10\% chance, the outcome generated is 0 or 1 with equal probability and is thus entirely uninformative.\footnote{The observation continues to be the final outcome. It does not include information on whether the outcome comes from the clinical trial or the roulette lottery.}
Clearly, $(\{0,1\},\pi)$ dominates $(\{0,1\},\pi')$ in the sense of Blackwell's informativeness, and thus also in the sense of elicitation. The corresponding Markov matrix is
\[
M = \Matrix{.95,.05;.05,.95}.
\]
Observe that $(\{0,1\},\pi')$ dominates $(\{0,1\},\pi)$ in the sense of elicitation, because $\pi = \pi' M^{-1}$, while it is not more informative than $(\{0,1\},\pi)$---and indeed it is easily verified that $M^{-1}$ is not Markov.

This fact turns out to be quite general: The main difference between elicitation and Blackwell informativeness is that, when data is used in contingent payments as incentive generator, certain types of noise in the data do not impact the ability to offer incentives, whereas in a context of estimation, the statistician cares to avoid every type of noise. Of course, the incentive designer is not immune to all noises. In the example above, no information can be elicited with the experiment that, with 50\% chance, reveals the true result of the clinical trial, and with the complementary probability reveals the opposite of the true result.

To formalize this idea, we introduce the concept of uniform garbling. An experiment $(Y,\pi')$ is a \df{uniform garbling} of another experiment $(Y,\pi)$ that shares the same outcome space when, for some $\epsilon \in [0,1)$, $(Y,\pi')$ reveals the same outcome\footnote{Strictly speaking, outcomes are always assumed to be conditionally independent across experiments, so by `same' outcome we mean an outcome with an identical distribution conditionally on the parameter.} as  $(Y,\pi)$ with probability $1-\epsilon$, and reveals an outcome drawn uniformly from $Y$ with the complementary probability: $\pi'(y|\theta) = \epsilon/|Y| + (1-\epsilon) \pi(y|\theta)$.
We then prove the following result:
\begin{proposition}
	\label{prop:comparison}
	An experiment $(Y,\pi_Y)$ dominates another experiment $(Z,\pi_Z)$ in the sense of elicitation if, and only if, $(Y,\pi_Y)$ is more informative than a uniform garbling of $(Z,\pi_Z)$.
\end{proposition}
Importantly, this result states that the elicitation order is simply the transitive closure of Blackwell's informativeness order and the order induced by uniform garblings. We may view the elicitation order as refining the Blackwell order with a lot of indifference.
The proof of Proposition~\ref{prop:comparison} is in Appendix~\ref{appx:Proofs of Comparison of experiments}.

The uniform distribution is enough to obtain the result, but not at all essential.
For example, fix an arbitrary distribution $\mu$ on every outcome space, and look at $\mu$-garblings instead, where $(Y,\pi')$ is a $\mu$-garbling of $(Y,\pi)$ if $\pi'$ is as in the case of a uniform garbling except that, instead of the uniform draw, the outcome is drawn according to $\mu$. Then, $\pi' = \pi M$ with $M$ the Markov matrix defined by
\[
M(z|y) =
\begin{cases}
1-\epsilon & \text{if } y=z, \\
\epsilon \mu(z) & \text{otherwise,}
\end{cases}
\]
and $M$ is invertible.\footnote{Observe that $M = (1-\epsilon) I + \epsilon P$ where $I$ is the identity matrix, and $P$ is a square Markov matrix whose rows are identical. Since $P$ is idempotent, it is easily verified that the inverse is $(1-\epsilon)^{-1} (I - \epsilon P)$.} 
Hence, $(Y,\pi)$ and $(Y,\pi')$ elicit the same information, and it can be verified that the arguments of Proposition~\ref{prop:comparison} continue to hold when $\mu$-garblings are used in place of uniform garblings.
\footnote{Also note that uniform garblings (or $\mu$-garblings) are not the only types of noisy transformations the principal is immune to, since any experiment $(Y,\pi')$ that is a garbling of $(Y,\pi)$, and so satisfies $\pi'=\pi M$ for a Markov matrix $M$, elicits the same information as $(Y,\pi)$ as long as $M$ has full rank.}


\subsection{Comparison based on the Power of Incentives}
\label{subsec:Comparison based on Incentives}

The order just defined does not deal with the power of incentives. We may wish that, in addition to the ability to elicit more information, the experiment also enables us to keep the same incentives; that is, we may wish to make comparisons based on incentives rather than just the information elicited.
To make this comparison, let us say that a (general) mechanism $\phi:\scR \times Y \to \R$ for an experiment $(Y,\pi_Y)$ is \df{payoff-equivalent} to another mechanism $\psi:\scR \times Z \to \R$ for an experiment $(Z,\pi_Z)$ if, for all beliefs $p \in \Delta(\Theta)$, and all reports $r \in \scR$, we have $\E_p[\phi(r,y)] = \E_p[\psi(r,z)]$.

An `incentive' order can then be defined, according to which an experiment $(Y,\pi_Y)$ dominates another, $(Z,\pi_Z)$, if the incentives that can be implemented using $(Z,\pi_Z)$ can also be implemented using $(Y,\pi_Y)$. Specifically, to any mechanism for $(Z,\pi_Z)$, there corresponds a payoff-equivalent mechanism for $(Y,\pi_Y)$.
This comparison demands that not only we can elicit the same information with the dominating experiment, but we can do so with the same incentives.
It turns out that this incentive order coincides with the elicitation order.

Indeed, recall that if $(Y,\pi_Y)$ dominates $(Z,\pi_Z)$ in the sense of elicitation, then there exists a $Y$-by-$Z$ matrix $M$ with $\pi_Z = \pi_Y M$.
Consider a mechanism $\psi$ for $(Z,\pi_Z)$ and construct a mechanism $\phi$ for $(Y,\pi_Y)$ by
\begin{equation}
\label{eq:construction-M}
\phi(r,y) = \sum_z \psi(r,z) M(y,z).
\end{equation}
For every belief over parameters, $p \in \Delta(\Theta)$,
\begin{align*}
\E_p[\phi(r,y)]
&= \int_\Theta \sum_y \sum_z \psi(r,z) \cdot M(y,z) \cdot \pi_Y(y|\theta) \cdot p(\de\theta) \\
&= \int_\Theta \sum_z \psi(r,z) \cdot \pi_Z(z|\theta) \cdot p(\de\theta)\\
&= \E_p[\psi(r,z)],
\end{align*}
hence, $\phi$ and $\psi$ are payoff-equivalent, which proves the next proposition:
\begin{proposition}
\label{prop:power-general}
If $(Y,\pi_Y)$ dominates $(Z,\pi_Z)$ in the sense of elicitation, then given any mechanism for $(Z,\pi_Z)$, there exists a payoff-equivalent mechanism for $(Y,\pi_Y)$.
\end{proposition}

In particular, whenever an experiment strictly dominates another in the sense of elicitation, a strictly richer set of incentives can be implemented. However, if we normalize the payoffs of the truthful analyst, it is not possible to use this flexibility to augment the power of incentives by a stronger punishment of deviations.
Returning to Example~\ref{ex:bernoulli model}, with a single clinical trial we elicit the mean parameter value with the incentive compatible mechanism $\phi(p,y) = 2 \mu y - \mu^2$ where $\mu = \E_p[\theta]$ is the mean parameter according to the reported parameter distribution $p$.
The analyst who is offered this mechanism and reports his belief $p$ truthfully earns on average $\E_p[\theta]^2$, and the cost of deviating to some other report $q$ is $(\E_p[\theta] - \E_q[\theta])^2$. 
If instead we observe the outcomes of 10 clinical trials and continue to elicit the mean parameter, we cannot leverage the additional data to generate a higher cost of deviation if we insist on paying on average $\E_p[\theta]^2$ to the truthful analyst.
This observation is formalized in the proposition below. We call a direct mechanism $\phi$ continuous if $p \mapsto \phi(p,y)$ is continuous in the total variation distance.
\begin{proposition}
\label{prop:power-valuefunction}
Consider two continuous direct incentive compatible mechanisms, $\phi$ and $\psi$ (either for one same experiment or for two different experiments).
If $\phi$ and $\psi$ generate the same expected payoffs to the truthful analyst, then the expected payoff of an analyst who believes $p$ and reports $q$ is identical under both mechanisms.
\end{proposition}
This result owes to the Envelope Theorem. A short proof is in Appendix~\ref{appx:Proofs of Comparison of experiments}.

While it is always possible to maintain expected payoffs with a dominating experiment, the range of the possible payments may need to be enlarged. Intuitively, it may be necessary to compensate for the presence of additional noise in the data that the dominating experiment generates.


In particular, if the principal has a limited liability constraint, so that the payments to the agent have to be nonnegative, then domination in the sense of elicitation does not imply in general that expected payments can be maintained. This can be easily seen in our example of Section~\ref{subsec:Comparison based on Elicitability}: $(\{0,1\},\pi')$ dominates $(\{0,1\},\pi)$ in the sense of elicitation, but consider a mechanism $\psi$ for the $\pi$ experiment in which one of the reports $r$ is such that $\psi(r,1)=1$ and $\psi(r,0)=0$. For any belief $p$ we have that $\E_p[\psi(r,z)]=\E_p[\theta]$, so in particular $\E_p[\psi(r,z)]=0$ if $p$ assigns probability 1 to state $\theta=0$. If a $\pi'$-mechanism $\phi$ is to be payoff-equivalent to $\psi$, then for this belief $p$ we must have $0=\E_p[\phi(r,y)] = .05\phi(r,1)+.95\phi(r,0)$, which by nonnegativity implies that $\phi(r,1)=\phi(r,0)=0$. But then $\E_p[\phi(r,y)] \neq \E_p[\psi(r,z)]$ for any other belief $p$.  

In our next result we characterize when it is possible to maintain expected payments, and hence incentives, when limited liability is required.

\begin{proposition}
\label{prop:power-nonnegative}
Let $(Y,\pi_Y)$ and $(Z,\pi_Z)$ be two experiments. Then the following two statements are equivalent:\\
(i) Given any mechanism for $(Z,\pi_Z)$ with nonnegative payoffs there exists a payoff-equivalent mechanism for $(Y,\pi_Y)$ also with nonnegative payoffs.\\
(ii) There exists a nonnegative $Y$-by-$Z$ matrix $M$ such that $\pi_Z = \pi_Y M$.
\end{proposition}

Notice how condition (ii) of the proposition, that $\pi_Z = \pi_Y M$ for some nonnegative $M$, is situated between the conditions that characterize the elicitation order defined in Section~\ref{subsec:Comparison based on Elicitability} and Blackwell's informativeness order.\footnote{The condition $\pi_Z = \pi_Y M$ for some nonnegative $M$ also appears in the work of \citet{LehrerShmaya2008}. There, it is used to characterize the case where $\pi_Y$ gives positive expected payoff whenever $\pi_Z$ does in the presence of an outside option.} In the former $M$ can be any matrix, while in the latter it needs to be a Markov matrix. We have already demonstrated with the above example that condition (ii) is strictly more demanding than the former; we now give an example showing that it is strictly less demanding than the latter.

Let $\Theta=\{\theta_1,\theta_2,\theta_3\}$, $Y=\{1,2,3,4\}$, $Z=\{1,2,3\}$.
Consider the following Markov kernels in their matrix representations:
\[
\pi_Y = \Matrix{.5,0,0,.5;0,.5,0,.5;0,0,.5,.5}
\qquad \text{and} \qquad
\pi_Z = \Matrix{.5,.5,0;.5,0,.5;0,.5,.5}.
\]
Observe that $(Y,\pi_Y)$ is not more informative than $(Z,\pi_Z)$, yet $\pi_Z = \pi_Y M$ with
\[
M = \Matrix{1,1,0;1,0,1;0,1,1;0,0,0}.
\]

However, when the dominating experiment is complete---as defined in Section~\ref{sec:model}---the limited liability constraints become stronger and the incentive order reduces precisely to Blackwell's informativeness order. This is the content of the next corollary.

\begin{corollary}
\label{coro:power-blackwell-1}
Let $(Y,\pi_Y)$ and $(Z,\pi_Z)$ be two experiments, and suppose that $(Y,\pi_Y)$ is complete. Then the following two statements are equivalent:\\
(i) Given any mechanism for $(Z,\pi_Z)$ with nonnegative payoffs there exists a payoff-equivalent mechanism for $(Y,\pi_Y)$ also with nonnegative payoffs.\\
(ii) $(Y,\pi_Y)$ is more informative than $(Z,\pi_Z)$.
\end{corollary}
The proof of Corollary~\ref{coro:power-blackwell-1} is in Appendix~\ref{appx:Proofs of Comparison of experiments}.

We now add the further restriction that the range of payments in the payoff-equivalent mechanism for the dominating experiment be contained in the range of payments of the original mechanism. This restriction is relevant when mechanism payoffs are not final payments but rather performance scores taking values on an exogenous scale, or probabilities of getting a fixed reward. In the latter case, realized payoffs are associated to expected payments, or expected utilities. A well-known property of these mechanisms, often used in experimental economics, is to solicit the truth from general expected utility maximizers, independently of the underlying utility function.\footnote{This point is made in \citet{savage1971elicitation} and exploited notably by \citet{roth1979game}, \citet{grether1981financial} and \citet{karni2009mechanism}.}

Notice first that if $(Y,\pi_Y)$ is more informative than $(Z,\pi_Z)$, so that $\pi_Z = \pi_Y M$ for some Markov matrix $M$, then the payments $\phi(r,y)$ defined in Equation~\eqref{eq:construction-M} are convex combinations of $\psi(r,z)$, $z \in Z$, and hence, the range of the possible payments is no larger under $\phi$ than under $\psi$. Therefore, being more informative is a sufficient condition for being able to generate the same incentives without increasing the range of payments; surprisingly, it is not a necessary condition.

To see why, let us go back to the last example with 
\[
\pi_Y = \Matrix{.5,0,0,.5;0,.5,0,.5;0,0,.5,.5}
\qquad \text{and} \qquad
\pi_Z = \Matrix{.5,.5,0;.5,0,.5;0,.5,.5}.
\]
Recall that $\pi_Y$ is not more informative than $\pi_Z$. Consider any mechanism $\psi(r,z)$ with payoffs in $[0,1]$. We now define a payoff-equivalent mechanism $\phi(r,y)$ also with payoffs in $[0,1]$. Fix $r$ and let us write $\psi_z$ instead of $\psi(r,z)$ for any $z\in Z$, and $\phi_y$ instead of $\phi(r,y)$ for any $y\in Y$. By symmetry, we may assume without loss that $\psi_1\ge \psi_2 \ge \psi_3$. Define $\phi_4=0$ if $\psi_1+\psi_2\le 1$, and $\phi_4=\psi_1+\psi_2-1$ otherwise. Next, define $\phi_1=\psi_1+\psi_2-\phi_4$, $\phi_2=\psi_1+\psi_3-\phi_4$, and $\phi_3=\psi_2+\psi_3-\phi_4$.  

We leave it to the interested reader to verify that $0\le \phi_y \le 1$ for every $y\in Y$. To see that $\phi(r,\cdot)$ gives the same expected payoff as $\psi(r,\cdot)$ for any belief $p$, suppose first that $p$ assigns probability 1 to state $\theta_1$. The expected payoff under $\psi$ is then $.5(\psi_1+\psi_2)$, while under $\phi$ it is $.5(\phi_1+\phi_4) = .5(\psi_1+\psi_2)$, i.e., they are equal. Similar calculations apply for states $\theta_2$ and $\theta_3$, implying that expected payoffs are identical under any belief $p$.

\medskip

What property then characterizes the case where the range of payoffs in the dominating experiment need not increase to generate the same incentives? It turns out that the answer requires us to think about transitions from signals in the dominating experiment $\pi_Y$ to \emph{sets of signals} in the dominated experiment $\pi_Z$, instead of transitions from signals to signals as in Blackwell's order. We let $N$ stand for a $Y$-by-$2^Z$ matrix with rows corresponding to signals in $Y$ and columns corresponding to events (subsets) in $Z$. We then have the following. 

\begin{proposition}
\label{prop:power-blackwell-2}
Let $(Y,\pi_Y)$ and $(Z,\pi_Z)$ be two experiments. Then the following two statements are equivalent:\\
(i) Given any mechanism for $(Z,\pi_Z)$ with payoffs in $[0,1]$ there exists a payoff-equivalent mechanism for $(Y,\pi_Y)$ also with payoffs in $[0,1]$.\\
(ii) There exists a $Y$-by-$2^Z$ matrix $N$ with entries in $[0,1]$ such that, for every $A\subseteq 2^Z$ and every $\theta$,
\begin{equation}
\label{eq:signals-to-events}
\sum_{z\in A} \pi_Z(z|\theta) = \sum_y \pi_Y(y|\theta)N(y,A).
\end{equation}
\end{proposition}
The proof of Proposition~\ref{prop:power-blackwell-2} is in Appendix~\ref{appx:Proofs of Comparison of experiments}. Note that if $\pi_Y$ is more informative than $\pi_Z$, $\pi_Z = \pi_Y M$ for some Markov matrix $M$, then by defining $N(y,A)=\sum_{z\in A} m(z|y)$ we obtain the equality in (\ref{eq:signals-to-events}).

\newpage

\appendix

\section{Proofs of Section~\ref{sec:elicitable information}}
\label{appx:Proofs of Elicitable information}

\subsection{Proof of Theorem~\ref{thm:baseline}}

We first remark that there exists a sequence of measurable sets of outcomes, $E_1, E_2, \dots,$ such that outcome distributions are identified on the family; that is, for any probability measures $\mu, \nu$ on $Y$, $\mu = \nu$ if and only if $\mu(E_i) = \nu(E_i)$ for every $i \in \N_+$. Indeed, since $Y$ is a Polish space, it has a countable base for its topology. The Borel $\sigma$-algebra on $Y$ is thus generated by a countable collection of events $\{ A_1, A_2, \dots\}$. For each $n$, the algebra $G_n$ generated by the truncated collection $\{A_1, \dots, A_n\}$ is finite, so the algebra $G$ generated by the full collection $\{ A_1, A_2, \dots\}$, which is equal to the union of all $G_n$, is countable. By the Carath\'{e}odory extension theorem, any probability measure on $Y$ is uniquely defined by the probability assigned to every event on $G$.

Let $y$ be the random outcome generated by $(Y,\pi)$.
A belief $p$ on parameters induces a belief $\lambda_p$ on outcomes as follows:
\begin{equation*}
	\lambda_p(A) = \int_\Omega \pi(A | \omega) p(\de \omega) = \E_p[\pi(A|\omega)].
\end{equation*}

Aside from the reported belief, the only other input to any mechanism for $(Y,\pi)$ is the experiment outcome, so any information that the mechanism elicits must be coarser than the information associated with the full information on the outcome distribution induced by the belief, $\scP^\star$. 

We describe below a simple mechanism that elicits $\scP^\star$.

Fix an arbitrary full-support distribution over the positive integers and consider the following protocol. The analyst first reports a parameter distribution. Then, a positive integer $i$ is drawn at random from the full-support distribution, and the analyst is paid
\begin{equation}
	1 - \pp{\lambda_p(E_i) - \1_{E_i}(x)}^2,\label{eq:thmbaseline-quadratic}
\end{equation}
where $\1_{E}$ is the indicator function of the set $E$.

Suppose the analyst's true belief is captured by the parameter distribution $p$.
A parameter distribution $q$ is a best response if and only if the analyst maximizes the expected value of Equation~\eqref{eq:thmbaseline-quadratic} for every positive integer $i$---since every $i$ is drawn with positive probability and Equation~\eqref{eq:thmbaseline-quadratic} is the classical Brier score \citep{brier1950verification}. 
This is equivalent to having $\lambda_p(E_i) = \lambda_q(E_i)$ for every $i$, which, in turn, is equivalent to $\lambda_p = \lambda_q$ by the remark above.
Hence, any strict best response includes full information on the outcome distribution induced by the true belief.
Since the analyst maximizes expected payoffs, the randomized protocol just mentioned is equivalent to the deterministic mechanism $\phi$ defined as 
\begin{equation*}
\phi(p,x) = 1 - \sum_{i} \Pr[i] \cdot \pp{\lambda_p(E_i) - \1_{E_i}(x)}^2.
\end{equation*}
We conclude there exists a mechanism $\phi$ for experiment $(Y,\pi)$ that elicits $\scP^\star$.


\subsection{Proof of Corollary~\ref{cor:estimator}}

Consider the product space $\Theta \times Y$ with its product $\sigma$-algebra. 
Let $\mu_p$ be the probability measure over that space induced by a belief $p$ and the experiment $(Y,\pi)$: For measurable sets $A \subseteq \Theta$ and $B \subseteq Y$,
\begin{equation*}
	\mu_p(A \times B) = \int_A \pi(B|\theta) p(\de \theta).
\end{equation*}
We continue to denote by $\lambda_p$ the outcome distribution induced by belief $p$.
Let $g$ and $w$ be as in the statement of Corollary~\ref{cor:estimator}, and $(\theta,y)$ be the random pair (parameter,outcome) generated by $\mu_p$. By definition, $g(\theta)$ is (a version of) the conditional expectation of $w(y)$ given $\theta$, so by the law of iterated expectations,
\[
\E_p[g(\theta)] 
=
\int_\Theta g(\theta) p(\theta)
=
\int_{\Theta \times Y} w(y) \mu_p(\de (\theta, y))  
= 
\int_Y w(y) \lambda_p(\de y).
\]
For any two beliefs $p$ and $q$ that are indistinguishable under $\scP^\star$, $\lambda_p = \lambda_q$, so $\E_p[g(\theta)] = \E_q[g(\theta)]$ which implies that the mean of $g(\theta)$ describes a coarser information than $\scP^\star$. By Theorem~\ref{thm:baseline}, we conclude that the mean of $g(\theta)$ is elicitable with $(Y,\pi)$.


\subsection{Proof of Corollary~\ref{cor:estimator-converse}}

Let $g \in \scG$.
As $(Y,\pi)$ is categorical we can write $Y=\{y_1, \dots, y_n\}$. 
By Theorem~\ref{thm:baseline}, the maximal information $\scP^\star$ that $(X,\pi)$ elicits is the partition induced by the means of $g_1(\theta), \dots, g_n(\theta)$ defined by $g_i(\theta) = \pi(y_i | \theta)$.

Let $\scM$ be the vector space of finite signed measures on $\Omega$, and let $\Gamma_i$ be the linear functional on $\scM$ defined by
\[
	\Gamma_i(\mu) = \int_\Theta g_i(\theta) \mu(\de\theta),
\]
and $\Gamma$ be the linear functional defined by
\begin{equation*}
	\Gamma(\mu) = \int_\Theta g(\theta) \mu(\de\theta).
\end{equation*}

Let us show the implication
\begin{equation}
	\label{eq:implication_psimu}
	\Gamma_1(\mu)=0, \dots, \Gamma_n(\mu)=0 \quad\implies\quad \Gamma(\mu) = 0
\end{equation}
for all $\mu \in \scM$.

The implication is immediate if $\mu = 0$, so suppose $\mu \ne 0$.
By the Jordan decomposition theorem, $\mu = \mu^+ - \mu^-$ where $\mu^+,\mu^-$ are positive measures.
If $\int_\Theta g_i(\theta) \mu(\de\theta) = 0$ for all $i$ then $\mu^+(\Theta)=\mu^-(\Theta)$ because $\sum_i g_i = 1$, and so $\mu = \alpha p - \alpha q$ for $p,q \in \Delta(\Theta)$ and $\alpha = \mu^+(\Theta) = \mu^-(\Theta) > 0$.
Hence, $\Gamma_i(\mu)=0$ for all $i$ implies $\E_p[g_i(\theta)] =  \E_q[g_i(\theta)]$ for all $i$.
Since $\scP^\star$ is the most refined information partition elicitable with $(Y,\pi)$, if $\E_p[g_i(\theta)] =  \E_q[g_i(\theta)]$ for all $i$, then $\E_p[g(\theta)] = E_q[g(\theta)]$, and so $\Gamma(\mu)=0$. Therefore, Equation~\ref{eq:implication_psimu} holds for all $\mu \in \scM$.

By Equation~\ref{eq:implication_psimu} and the Fundamental Theorem of Duality (see, for example, Theorem 5.91 of \citealp{aliprantis2006infinite}), $\Gamma$ is in the linear span of $\Gamma_1, \dots, \Gamma_n$. Therefore, $\Gamma = \sum_i w(y_i) \Gamma_i$ for some function $w:Y \to \R$. Considering an arbitrary parameter value $\theta$ and letting $\mu_\theta$ be the Dirac measure centered on the point $\theta$, it follows that
\[
g(\theta) = \Gamma(\mu_\theta) = \sum_i w(y_i) \Gamma_i(\mu_\theta) = \sum_i w(y_i) g_i(\theta).
\]


\subsection{Proof of Theorem~\ref{thm:multiple}}

Recall that $\Theta$ is a Polish topological space endowed with the Borel $\sigma$-algebra, and $\scG$ is the set of bounded measurable real-valued functions on $\Theta$. We continue to denote by $\scM$ the set of finite signed measures on $\Theta$.

For $\mu \in \scM$ and $g \in \scG$, consider the bilinear functional $\langle g, \mu \rangle = \int_\Theta g(\theta) \mu(\de \theta)$. Note that with this bilinear functional, $\langle \scG, \scM \rangle$ forms a dual pair (in the sense of Definition 5.90 of \citealp{aliprantis2006infinite}). We endow $\scG$ with the weak topology associated with this dual pair, whereby the set of continuous linear functionals on $\scG$ coincides with $\scM$, in the sense that $\Gamma$ is a continuous linear functional on $\scG$ if and only if there exists $\mu \in \scM$ such that $\Gamma(g) = \langle g, \mu \rangle$ for all $g \in \scG$ \citep[Theorem 5.93]{aliprantis2006infinite}.

Theorem~\ref{thm:multiple} then follows from Lemmas~\ref{lem:multiple-1} and~\ref{lem:multiple-2} below.

In the first lemma, we consider an arbitrary experiment $(Y,\pi)$. We write $\scL$ for the linear span of the functions $\theta \mapsto \pi(A | \theta)$ with $A \subseteq Y$ a measurable set of outcomes, and write $\widebar \scL$ for the closure of $\scL$.
\begin{lemma}
	\label{lem:multiple-1}
	For any $g \in \scG$, the mean of $g(\theta)$ is elicitable with $(Y,\pi)$ if, and only if, $g \in \widebar \scL$.
\end{lemma}
\begin{proof}
We first prove that if $g \in \widebar \scL$ then the mean of $g(\theta)$ is elicitable with $(Y,\pi)$.

Suppose $p, q \in \Delta(\Theta)$ are such that $\E_p[g(\theta)] \ne \E_q[g(\theta)]$.
Then $\langle g, p - q \rangle \ne 0$. Because $f \mapsto \langle f, p - q \rangle$ is a continuous linear functional, there exists an open set $U$ that includes $g$ such that for all $f \in U$, $\langle f, p - q \rangle \ne 0$. Since $g$ is in the closure of $\scL$, there exists $f \in U$ such that $f \in \scL$. Hence, for this function $f$, $\E_p[f(\theta)] \ne \E_q[f(\theta)]$. By Theorem~\ref{thm:baseline}, the mean of $f(\theta)$ is elicitable with $(Y,\pi)$, thus $p$ and $q$ belong to different members of the maximal information partition $\scP^\star$ as defined in Section~\ref{sec:elicitable information}. Hence, the mean of $g(\theta)$ is elicitable.

We now prove the converse, that if the mean of $g(\theta)$ is elicitable, then $g \in \widebar \scL$.

Suppose, by contradiction, that $g \not\in \widebar \scL$.
The space $\scG$ being equipped with the weak topology, it is locally convex---a generalization of normed vector spaces---and by the separating hyperplane theorem for locally convex spaces \citep[Theorem 5.79]{aliprantis2006infinite}, there exists a continuous linear functional $\Gamma$ on $\scG$ such that $\Gamma(f)=0$ for all $f \in \widebar \scL$ while $\Gamma(g) \ne 0$.
Since the set of continuous linear functionals coincides with $\scM$, there is a finite signed measure $\mu$ such that for all $f \in \scG$,
\[
\Gamma(f) = \int_\Theta f(\theta) \mu(\de\theta).
\]

Constant functions trivially belong to $\scL$ so $\mu(\Theta)=0$. Then, by the Jordan decomposition theorem, we can write $\mu = \alpha p - \alpha q$, where $p,q \in \Delta(\Theta)$ and $\alpha > 0$.
That $\Gamma(g) \ne 0$ implies $\langle g, \alpha p - \alpha q \rangle \ne 0$, or equivalently, $\E_p[g(\theta)] \ne \E_q[g(\theta)]$. 
And since $\Gamma(f)=0$ when $f \in \widebar \scL$, we have $\E_p[f(\theta)] = \E_q[f(\theta)]$.

However, when $p$ and $q$ are beliefs such that for every measurable set $A \subseteq Y$, it is the case that $\E_p[\pi(A|\theta)] = \E_q[\pi(A|\theta)]$, then $p$ and $q$ belong to the same partition of the information partition $\scP^\star$ that, by Theorem~\ref{thm:baseline}, captures the maximal information elicited by $(Y,\pi)$, and so we must have $\E_p[g(\theta)] = \E_q[g(\theta)]$, hence a contradiction.
\end{proof}

In this second lemma, for any two linear subspaces of $\scG$, $V_1$ and $V_2$, we call pointwise product of $V_1$ and $V_2$ the set $\{g_1 g_2 : (g_1,g_2) \in V_1 \times V_2\}$ where $g_1 g_2$ simply refers to the pointwise product of $g_1$ and $g_2$.

\begin{lemma}
	\label{lem:multiple-2}
	Given two linear subspaces of $\scG$, the pointwise product of their closures is included in the closure of their pointwise products.
\end{lemma}
\begin{proof}
For $i=1,2$, let $V_i$ be a linear subspace of $\scG$ and $g_i$ be in the closure of $V_i$.

Let $\scB$ be the collection of sets of the form
\[
\pc{h \in \scG : \forall k,\ \abs{\int_\Theta \pp{h(\theta) - g_1(\theta) g_2(\theta)} \nu_k(\de \theta)} < \epsilon}
\]
for $\epsilon > 0$ and finitely many $\nu_1, \dots, \nu_K \in\scM$. Notice that the collection $\scB$ forms an open neighborhood base of $g_1 g_2$ by definition of the weak topology. Therefore, showing that $g_1 g_2$ is in the closure of the pointwise product of $V_1$ and $V_2$ reduces to showing that for every $B \in \scB$, there exists $(f_1,f_2) \in V_1 \times V_2$ such that $f_1 f_2 \in B$.

Fix a set $B \in \scB$. First, using the same notation as above and considering the signed measures $\eta_k$ defined as 
\[
	\eta_k(A) = \int_A g_1(\theta) \nu_k(\de\theta),
\]
we can choose $f_2 \in V_2$ such that for all $k=1, \dots, K$,
\[
	\abs{\langle (f_2 - g_2) g_1, \nu_k \rangle} 
	=
	\abs{\langle f_2 - g_2, \eta_k \rangle} 
	< 
	\frac{\epsilon}2,
\]
since $g_2$ is in the closure of $V_2$.
Second, considering the signed measures $\eta_k$ now defined instead as
\[
	\eta_k(A) = \int_A f_2(\theta) \nu_k(\de\theta),
\]
we can choose $f_1 \in V_1$ such that for all $k=1, \dots, K$,

\[
	\abs{\langle (f_1 - g_1) f_2, \nu_k \rangle} 
	=
	\abs{\langle f_1 - g_1, \eta_k \rangle} 
	< 
	\frac{\epsilon}2,
\]
since $g_1$ is in the closure of $V_1$.

Hence, for all $k=1,\dots,K$,
\[
	\abs{\langle f_1 f_2 - g_1 g_2, \nu_k \rangle}
	\le
	\abs{\langle (f_1 - g_1) f_2, \nu_k \rangle}
	+
	\abs{\langle (f_2 - g_2) g_2, \nu_k \rangle}
	<
	\epsilon,
\]
and thus $f_1 f_2 \in B$.
\end{proof}

We now conclude the proof of Theorem~\ref{thm:multiple}. We focus on the case $n=2$, the case $n > 2$ is an immediate generalization.
For $i=1,2$, let $\scL_i$ be the linear span of the functions $\theta \mapsto \pi_i(A|\theta)$ for $A$ a measurable subset of $Y_i$.
Under the assumption of the theorem, the mean of $g_i(\theta)$ is elicitable with $(Y_i,\pi_i)$, implying by Lemma~\ref{lem:multiple-1} that $g_i$ is in the closure of $\scL_i$.
Then, by Lemma~\ref{lem:multiple-2}, $g_1 g_2$ is in the closure of the pointwise product of $\scL_1$ of $\scL_2$, and so by Lemma~\ref{lem:multiple-1} again, the mean of $g_1(\theta) g_2(\theta)$ is elicitable with $(Y_1,\pi_1) \otimes (Y_2,\pi_2)$.


\subsection{Proof of Theorem~\ref{thm:covariates}}

For all $x \in X$, let $\psi(\cdot | x)$ be a mechanism for $(Y,\pi_x)$ that elicits $\scP(x)$ and whose payoffs take values in a bounded interval such as $[0,1]$ (the proof of Theorem~\ref{thm:baseline} includes an instance of such a mechanism). 
To elicit information with outcomes from the mixture experiment, we consider the compound mechanism $\phi$ defined by $\phi(p,(x,y)) = \psi(p,y|x)$. We show that this mechanism elicits $\scP_0$---that is, fixing any two beliefs on parameters, $p$ and $q$, that are distinguishable under $\scP_0$, we show
\[
\E_p[\phi(p,(x,y))] > \E_p[\phi(q,(x,y))].
\]

Let $X^\infty$ be the space of all infinite sequences in $X$. A generic element of $X^\infty$ is denoted $x^\infty = (x_1, x_2, \dots)$. (The case of finite sequences is identical and omitted.)
We abuse notation and use again the symbol $\mu$ for the probability measure over sequences whose elements are drawn independently and identically according to (the original) $\mu$.

Under the assumptions of Theorem~\ref{thm:covariates}, there exists $\scS \subseteq X^\infty$ with $\mu(\scS) > 0$ and such that, for every $(x_1, x_2, \dots) \in \scS$, the join of the partitions in $\{ \scP(x_i) \}$ is a refinement of $\scP_0$.

Observe that
\begin{multline*}
\E_p[\phi(q,(x,y))] 
=
\int_\Theta \int_{X \times Y} \phi(q,(x,y))\cdot \pi(\de (x,y) | \theta) \cdot p(\de \theta)
\\ 
=
\int_{X^\infty} 
\pp{
\sum_{i=1}^\infty
\frac{1}{2^i}
\int_\Theta \int_Y
\psi(q,y|x_i) \cdot
\pi_{x_i}(\de y|\theta) \cdot
p(\de\theta)
}
\mu(\de x^\infty) 	
\end{multline*}
and similarly if $p$ is used in place of $q$. 
For every $x$, by incentive compatibility of $\psi(\cdot | x)$,
\[
\int_\Theta \int_Y \pp{\psi(p,y | x) - \psi(q,y | x)} \cdot \pi_{x}(\de y|\theta) \cdot p(\de\theta) \ge 0,
\]
and hence,
\begin{multline*}
\E_p[\phi(p,(x,y))] - \E_p[\phi(q,(x,y))]
\\
\ge 
\int_{\scS} 
\pp{
\sum_{i=1}^\infty
\frac{1}{2^i}
\int_\Theta \int_Y
\pp{\psi(p,y|x_i) - \psi(q,y|x_i)} \cdot
\pi_{x_i}(\de y|\theta) \cdot
p(\de\theta)
}
\mu(\de x^\infty). 	
\end{multline*}

Since, for every $(x_1, x_2, \dots) \in \scS$, the join of the partitions in $\{ \scP(x_i) \}$ is a refinement of $\scP_0$, there exists $j$ such that
$\E_p[\psi(p,y|x_j)] 
>
\E_p[\psi(q,y|x_j)]  
$ (where the random element in the expectation is $y$ and not $x_i$). Hence,
\[
\sum_{i=1}^\infty
\frac{1}{2^i}
\int_\Theta \int_Y
\pp{\psi(p,y|x_i) - \psi(q,y|x_i)} \cdot
\pi_{x_i}(\de y|\theta) \cdot
p(\de\theta)
>0,
\]
and since $\mu(\scS)>0$,
\[
\int_{\scS} 
\pp{
\sum_{i=1}^\infty
\frac{1}{2^i}
\int_\Theta \int_Y
\pc{\psi(p,y|x_i) - \psi(q,y|x_i)}
\pi_{x_i}(\de y|\theta)
p(\de\theta)
}
\mu(\de x^\infty)
>
0,
\]
which implies $\E_p[\phi(p,(x,y))] - \E_p[\phi(q,(x,y))] > 0$.


\section{Proofs of Section~\ref{sec:comparison of experiments}}
\label{appx:Proofs of Comparison of experiments}

\subsection{Proof of Proposition~\ref{prop:comparison}}

Let us start with the `if' part of the proposition: Suppose $(Y,\pi_Y)$ is more informative than $(Z,\pi_Z')$ in the sense of Blackwell, with $(Z,\pi_Z')$ some uniform garbling of $(Z,\pi_Z)$.

As discussed in Section~\ref{sec:comparison of experiments}, $(Y,\pi_Y)$ also dominates $(Z,\pi_Z')$ in the sense of in the sense of elicitation, and by transitivity, it is enough to show that $(Z,\pi_Z')$ dominates $(Z,\pi_Z)$ in the sense of elicitation. For some $\epsilon \in [0,1)$, $\pi_Z' = ((1-\epsilon) I + \epsilon/|Z| J) \pi_Z$, where $I$ is the identity matrix and $J$ is the unit matrix---here both square matrices of dimension $|Z|$. It is immediate that the matrix $((1-\epsilon) I + \epsilon/|Z| J)$ is full rank. Thus, $(Z,\pi_Z')$ dominates $(Z,\pi_Z)$ in the sense of elicitation. 

The `only if' part of the proposition makes use of the following lemma.

\begin{lemma}
\label{lem:yaron}
$(Y,\pi_Y)$ dominates $(Z,\pi_Z)$ in the sense of elicitation if, and only if, there exists a $Y$-by-$Z$ matrix $M$ with $\sum_z M(y,z) =1$ such that $\pi_Z = \pi_Y M$.
\end{lemma}
\begin{proof}
We have already discussed the `if' part in Section~\ref{sec:comparison of experiments}.
For the `only if' part, recall that if $(Y,\pi_Y)$ dominates $(Z,\pi_Z)$ in the sense of elicitation, then $\pi_Z = \pi_Y M$ for some $Y$-by-$Z$ matrix $M$.
Let $M(y) = \sum_z M(y,z)$ and $\widetilde M$ by the $Y$-by-$Z$ matrix with entries $\widetilde M(y,z) = M(y,z) + 1 - M(y)$. Note that $\sum_z \widetilde M(y,z) = M(y) + 1 - M(y) = 1$. 
In addition, for all $\theta$,
\begin{equation*}
\sum_y \widetilde M(y,z) \pi_Y(y|\theta) = \pi_Z(z|\theta) + 1 - \sum_y M(y) \pi_Y(y|\theta) = \pi_Z(z|\theta),
\end{equation*}
because
\begin{equation*}
\sum_y M(y) \pi_Y(y|\theta) = \sum_{y,z} M(y,z) \pi_Y(y|\theta) = \sum_z \pi_Z(z|\theta) = 1,
\end{equation*}
and hence, $\pi_Z = \pi_Y \widetilde M$.
\end{proof}

Suppose $(Y,\pi_Y)$ dominates $(Z,\pi_Z)$ in the sense of elicitation.
Let $(Z, \pi_U)$ be the degenerate experiment that gives an outcome from $Z$ uniformly at random, and let $T = (1-\epsilon) M + \epsilon/|Z| J$ for $0 < \epsilon < 1$, where $J$ is now the unit matrix of dimension $|Y|$-by-$|Z|$. 
Notice that $T$ has all positive entries if $\epsilon$ is close enough to 1, and notice that all rows of $T$ sum to 1, hence, we can choose $T$ to be Markov.
For such $\epsilon$, we consider the uniform garbling of $(Z,\pi_Z)$, $(Z,\pi_Z')$, for which the probability of the uniformly drawn outcome is $\epsilon$. Then, we have
\begin{align*}
\pi_Z' 
&= (1-\epsilon) \pi_Z   +  \epsilon \pi_U   \\
&= (1-\epsilon) \pi_Y M +  \epsilon \pi_Y J \\
&= \pi_Y \pp{ (1-\epsilon) M + \epsilon J}  \\
&= \pi_Y T,
\end{align*}
and hence, $(Y,\pi_Y)$ is more informative than $(Z,\pi_Z')$.


\subsection{Proof of Proposition~\ref{prop:power-valuefunction}}

Fix a mechanism $\phi:\Delta(\Theta) \times Y \to \R$ for $(Y,\pi_Y)$, and another mechanism $\psi:\Delta(\Theta) \times Z \to \R$ for $(Z,\pi_Z)$, both continuous and incentive compatible.

We let $V_\phi$ be the value function of $\phi$,
\[
V_\phi(p) = \sup_{q \in \Delta(\Theta)} \E_p[\phi(q,y)],
\]
and similarly for $\psi$. If the expected payoffs of the truthful analyst are the same under both mechanisms, then $V_\phi = V_\psi$.

Consider $V:[0,1] \to \R$ defined as
\[
V(\alpha) 
= V_\phi(\alpha p + (1-\alpha)q) 
= \sup_{\alpha'} \E_{\alpha p + (1-\alpha)q} \pb{ \phi(\alpha' p + (1-\alpha')q, y)}.
\]
By the envelope theorem \citet[Theorem 2]{milgrom2002envelope}, $V$ is differentiable and 
\[
V'(\alpha) = \E_p\pb{\phi(\alpha p + (1-\alpha)q, y)} - \E_q\pb{\phi(\alpha p + (1-\alpha)q, y)}.
\]
The continuity of $\phi$ also makes $V'$ continuous, so
\[
V'(0) = \E_p\pb{\phi(q, y)} - \E_q\pb{\phi(q, y)},
\]
and hence,
\[
\E_p\pb{\phi(q, y)} = V_\phi(q) + \frac{\partial}{\partial \alpha} V_\phi(\alpha p + (1-\alpha)q).
\]
Since $V_\phi = V_\psi$, it follows that $\E_p\pb{\phi(q, y)} = \E_p\pb{\psi(q, z)}$.

\subsection{Proof of Proposition~\ref{prop:power-nonnegative}}

$(i) \Longrightarrow (ii)$\\
Consider a nonnegative mechanism $\psi$ for $\pi_Z$ where the set of possible reports is $\scR=Z$, and such that $\psi(z,z')=1$ if $z=z'$ and $\psi(z,z')=0$ otherwise. By assumption, there is a nonnegative mechanism $\phi$ for $\pi_Y$ with the same set of reports $\scR=Z$ such that $\sum_y \phi(z,y) \pi_Y(y|\theta) = \sum_{z'} \psi(z,z) \pi_Z(z'|\theta) = \pi_Z(z|\theta)$ for every state $\theta$. Define $M(y,z)=\phi(z,y)\ge 0$ to get the required matrix $M$.

$(ii) \Longrightarrow (i)$\\
Given a nonnegative mechanism $\psi$ for $\pi_Z$, define a mechanism $\phi$ for $\pi_Y$ as in (\ref{eq:construction-M}). Since $M$ is nonnegative, so is $\phi$.


\subsection{Proof of Corollary~\ref{coro:power-blackwell-1}}

For $p \in \Delta(\Theta)$, let $\lambda_p$ be the distribution over $Z$ induced by the belief $p$: for all $z \in Z$,
\[
\lambda_p(z) = \int_\Theta \pi(y|\theta) p(\de \theta).
\]

Observe that by the condition in Corollary~\ref{coro:power-blackwell-1}, $(Y,\pi_Y)$ dominates $(Z,\pi_Z)$ in the sense of elicitation (for example, one may apply the condition on the mechanism constructed in the proof of Theorem~\ref{thm:baseline}). Therefore, as discussed in Section~\ref{sec:comparison of experiments}, there exists a $Y$-by-$Z$ matrix $M$ such that $\pi_Z = \pi_Y M$, and by Lemma~\ref{lem:yaron} from the proof of Proposition~\ref{prop:power-general}, we choose $M$ such that every row sums to 1.

The null space of $\pi_Y$ is defined as
\[
\ker \pi_Y = \pc{ v \in \R^Y : \forall \theta \in \Theta,\ \sum_y  v(y) \pi(y|\theta)= 0}.
\]
Because $(Y,\pi_Y)$  is complete,  $\ker \pi_Y = \{ 0 \}$. Indeed, since $\R^Y$ is the linear span of the $(|Y|-1)$-simplex (identified with $\Delta(Y)$), it suffices to show that for any $v$ in the simplex, 
\[
\forall \theta \in \Theta,\ \sum_y  v(y) \pi(y|\theta)= 0
\qquad \implies \qquad
v = 0.
\]
This implication holds because, if $p \in \Delta(\Theta)$ is such that $\lambda_p = v$---which is possible because $(Y,\pi_Y)$  is complete---then,
\[
\int_\Theta \sum_y  v(y) \pi(y|\theta) p(\de\theta) = \sum_y v(y) \lambda_p(y) = 0
\qquad \implies \qquad
v = 0.
\]

Let $Z=\{z_1, \dots, z_n\}$, and let $\delta_i \in \Delta(\Theta)$ be the distribution over parameters whose induced distribution over $Z$, $\lambda_{\delta_i}$, puts full mass on $z_i$. The completeness of $(Z\,\pi_Z)$ ensures that this distribution exists.

Let $\psi$ be the mechanism 
\[
\psi(p,z)=\sum_{i} \pp{1 - \pp{\lambda_p(z_i) - \1\{z=z_i\}}^2}.
\]
Note that $\psi$ is direct and incentive-compatible, with nonnegative values.
Notice that for all $i \ne j$, $\psi(\delta_i,z_i)=1$ and $\psi(\delta_i,z_j)=0$.

By the assumption of Corollary~\ref{coro:power-blackwell-1} , there exists a direct, incentive-compatible mechanism $\phi$ with nonnegative payoffs such that for all $p \in \Delta(\Theta)$, and all $i$,
\[
\E_p[\phi(\delta_i,y)] = \E_p[\psi(\delta_i,z)].
\]
In particular, taking for $p$ the parameter distribution that puts full mass on $\theta$, we get
\[
\sum_y \phi(\delta_i,y) \pi_Y(y|\theta) 
=
\sum_z \psi(\delta_i,z) \pi_Z(z|\theta)
= 
\pi_Z(z_i|\theta)
= 
\sum_y M(y,z_i) \pi_Y(y|\theta).
\]
So
\[
\phi(\delta_i,\cdot) - M(\cdot,z_i) \in \ker \pi_Y = \{0\},
\]
which means that all entries of $M$ are nonnegative.

Since $\sum_z M(y,z) = 1$ for all $y$, $M$ is a Markov matrix, and thus $(Y,\pi_Y)$ is more informative than $(Z,\pi_Z)$.

\subsection{Proof of Proposition~\ref{prop:power-blackwell-2}}

$(i) \Longrightarrow (ii)$\\
Consider a mechanism $\psi$ for $\pi_Z$ with set of reports $\scR=2^Z$ defined by $\psi(A,z)=1$ if $z\in A$ and $\psi(A,z)=0$ otherwise. Since the payoffs of $\psi$ are in $[0,1]$, there is by assumption a payoff-equivalent mechanism $\phi(A,y)$ for $\pi_Y$ with payoffs also in $[0,1]$. For every state $\theta$ and every $A\in 2^Z$ we have that
$$\sum_{z\in A} \pi_Z(z|\theta) = \sum_z \pi_Z(z|\theta) \psi(A,z) = \sum_y \pi_Y(y|\theta) \phi(A,y),$$
where the first equality is by construction of $\psi$, and the second by payoff-equivalence for the belief $p$ that assigns probability 1 to state $\theta$. Defining $N(y,A)=\phi(A,y)$ completes the proof.

$(ii) \Longrightarrow (i)$\\
Consider any mechanism $\psi$ for $\pi_Z$ with payoffs in $[0,1]$, and we will construct a payoff-equivalent mechanism $\phi$ for $\pi_Y$ also with payoffs in $[0,1]$. Let $r\in\scR$ be any report, and consider the vector $\psi(r,\cdot)\in [0,1]^Z$. This vector can be represented as a convex combination of the extreme point of the hypercube $[0,1]^Z$, which are the indicators $\1_A(\cdot)$ of all subsets $A\in 2^Z$. Thus, there are non-negative numbers $\{\eta_r(A)\}_{A\in 2^Z}$ summing up to 1 such that for all $z$, $\psi(r,z) = \sum_{A\in 2^Z} \eta_r(A) \1_A(z)$. Let $\phi$ be defined by $\phi(r,y) = \sum_{A\in 2^Z} \eta_r(A) N(y,A)$, and note that $\phi(r,y)\in [0,1]$ as a convex combination of numbers in $[0,1]$. Also, for every state $\theta$ we have
\begin{eqnarray*}
\sum_y \pi_Y(y|\theta)\phi(r,y) = \sum_y \pi_Y(y|\theta) \sum_{A\in 2^Z} \eta_r(A) N(y,A) = \sum_{A\in 2^Z} \eta_r(A) \sum_y \pi_Y(y|\theta) N(y,A) =\\ 
\sum_{A\in 2^Z} \eta_r(A) \sum_{z\in A} \pi_Z(z|\theta) = \sum_z \pi_Z(z|\theta) \sum_{\{A: z\in A\}} \eta_r(A) = \sum_z \pi_Z(z|\theta) \psi(r,z),
 \end{eqnarray*}
where the first equality is by construction of $\phi$, the second is just a change in the order of summation, the third is by the assumption of the proposition, the fourth is again a change in the order of summation, and the last is by $\psi(r,z) = \sum_{A\in 2^Z} \eta_r(A) \1_A(z)$. This shows that $\phi$ gives the same expected payoff as $\psi$ at any state $\theta$ implying that they are payoff-equivalent.


\newpage
\bibliographystyle{chicago}
\bibliography{_references.bib}

\newcommand{\noop}[1]{}
\begin{thebibliography}{}

\bibitem[\protect\citeauthoryear{Abernethy and Frongillo}{Abernethy and
  Frongillo}{2012}]{abernethy2012}
Abernethy, J. and R.~Frongillo (2012).
\newblock A characterization of scoring rules for linear properties.
\newblock In {\em the 25th Annual Conference on Learning Theory (COLT)}.

\bibitem[\protect\citeauthoryear{Al-Najjar, Sandroni, Smorodinsky, and
  Weinstein}{Al-Najjar et~al.}{2010}]{al2010}
Al-Najjar, N., A.~Sandroni, R.~Smorodinsky, and J.~Weinstein (2010).
\newblock Testing theories with learnable and predictive representations.
\newblock {\em Journal of Economic Theory\/}~{\em 145\/}(6), 2203--2217.

\bibitem[\protect\citeauthoryear{Aliprantis and Border}{Aliprantis and
  Border}{2006}]{aliprantis2006infinite}
Aliprantis, C.~D. and K.~C. Border (2006).
\newblock {\em Infinite Dimensional Analysis: {A} Hitchhiker's Guide\/} (third
  ed.).
\newblock Springer.

\bibitem[\protect\citeauthoryear{Aumann}{Aumann}{1976}]{aumann1976agreeing}
Aumann, R.~J. (1976).
\newblock Agreeing to disagree.
\newblock {\em The Annals of Statistics\/}~{\em 4\/}(6), 1236--1239.

\bibitem[\protect\citeauthoryear{Barelli}{Barelli}{2009}]{barelli2009genericity}
Barelli, P. (2009).
\newblock On the genericity of full surplus extraction in mechanism design.
\newblock {\em Journal of Economic Theory\/}~{\em 144\/}(3), 1320--1332.

\bibitem[\protect\citeauthoryear{Blackwell}{Blackwell}{1951}]{blackwell1951comparison}
Blackwell, D. (1951).
\newblock The comparison of experiments.
\newblock In J.~Neyman (Ed.), {\em Proceedings of the Second Berkeley Symposium
  on Mathematical Statistics and Probability}, pp.\  93--102. University of
  California Press, Berkeley.

\bibitem[\protect\citeauthoryear{Brier}{Brier}{1950}]{brier1950verification}
Brier, G.~W. (1950).
\newblock Verification of forecasts expressed in terms of probability.
\newblock {\em Monthly Weather Review\/}~{\em 78\/}(1), 1--3.

\bibitem[\protect\citeauthoryear{Bruni and Koch}{Bruni and
  Koch}{1985}]{bruni1985identifiability}
Bruni, C. and G.~Koch (1985).
\newblock Identifiability of continuous mixtures of unknown {G}aussian
  distributions.
\newblock {\em The Annals of Probability\/}~{\em 13\/}(4), 1341--1357.

\bibitem[\protect\citeauthoryear{Chambers and Lambert}{Chambers and
  Lambert}{2021}]{chambers2021dynamic}
Chambers, C.~P. and N.~S. Lambert (2021).
\newblock Dynamic belief elicitation.
\newblock {\em Econometrica\/}~{\em 89\/}(1), 375--414.

\bibitem[\protect\citeauthoryear{Chen and Xiong}{Chen and
  Xiong}{2013}]{chen2013genericity}
Chen, Y.-C. and S.~Xiong (2013).
\newblock Genericity and robustness of full surplus extraction.
\newblock {\em Econometrica\/}~{\em 81\/}(2), 825--847.

\bibitem[\protect\citeauthoryear{Cram{\'e}r and Wold}{Cram{\'e}r and
  Wold}{1936}]{cramer1936theorems}
Cram{\'e}r, H. and H.~Wold (1936).
\newblock Some theorems on distribution functions.
\newblock {\em Journal of the London Mathematical Society\/}~{\em 1\/}(4),
  290--294.

\bibitem[\protect\citeauthoryear{Cr\'{e}mer and McLean}{Cr\'{e}mer and
  McLean}{1988}]{cremer1988full}
Cr\'{e}mer, J. and R.~P. McLean (1988).
\newblock Full extraction of the surplus in {B}ayesian and dominant strategy
  auctions.
\newblock {\em Econometrica\/}~{\em 56\/}(6), 1247--1257.

\bibitem[\protect\citeauthoryear{Dawid}{Dawid}{1982}]{dawid1982}
Dawid, A. (1982).
\newblock The well-calibrated {B}ayesian.
\newblock {\em Journal of the American Statistical Association\/}~{\em
  77\/}(379), 605--610.

\bibitem[\protect\citeauthoryear{De~Finetti}{De~Finetti}{1962}]{definetti1962}
De~Finetti, B. (1962).
\newblock Does it make sense to speak of ``good probability appraisers''?
\newblock In I.~Good (Ed.), {\em The Scientist Speculates: {A}n Anthology of
  Partly-Baked Ideas}, pp.\  257--364. Basic Books.

\bibitem[\protect\citeauthoryear{Foster and Vohra}{Foster and
  Vohra}{1998}]{foster1998}
Foster, D. and R.~Vohra (1998).
\newblock Asymptotic calibration.
\newblock {\em Biometrika\/}~{\em 85\/}(2), 379--390.

\bibitem[\protect\citeauthoryear{Frongillo and Kash}{Frongillo and
  Kash}{2015}]{frongillo2015a}
Frongillo, R. and I.~Kash (2015).
\newblock On elicitation complexity.
\newblock In {\em Advances in Neural Information Processing Systems 28 (NIPS)}.

\bibitem[\protect\citeauthoryear{Fu, Haghpanah, Hartline, and Kleinberg}{Fu
  et~al.}{2021}]{fu2021full}
Fu, H., N.~Haghpanah, J.~Hartline, and R.~Kleinberg (2021).
\newblock Full surplus extraction from samples.
\newblock {\em Journal of Economic Theory\/}~{\em 193}, 105230.

\bibitem[\protect\citeauthoryear{Ganzburg}{Ganzburg}{1981}]{ganzburg1981multidimensional}
Ganzburg, M.~I. (1981).
\newblock Multidimensional jackson theorems.
\newblock {\em Siberian Mathematical Journal\/}~{\em 22\/}(2), 223--231.

\bibitem[\protect\citeauthoryear{George and Buyse}{George and
  Buyse}{2015}]{george2015data}
George, S.~L. and M.~Buyse (2015).
\newblock Data fraud in clinical trials.
\newblock {\em Clinical Investigation\/}~{\em 5\/}(2), 161.

\bibitem[\protect\citeauthoryear{Gneiting}{Gneiting}{2011}]{gneiting2011}
Gneiting, T. (2011).
\newblock Making and evaluating point forecasts.
\newblock {\em Journal of the American Statistical Association\/}~{\em
  106\/}(494), 746--762.

\bibitem[\protect\citeauthoryear{Gneiting and Raftery}{Gneiting and
  Raftery}{2007}]{gneiting2007strictly}
Gneiting, T. and A.~E. Raftery (2007).
\newblock Strictly proper scoring rules, prediction, and estimation.
\newblock {\em Journal of the American Statistical Association\/}~{\em
  102\/}(477), 359--378.

\bibitem[\protect\citeauthoryear{Good}{Good}{1952}]{good1952rational}
Good, I.~J. (1952).
\newblock Rational decisions.
\newblock {\em Journal of the Royal Statistical Society, Series B\/}~{\em
  14\/}(1), 107--114.

\bibitem[\protect\citeauthoryear{Goodman}{Goodman}{1954}]{goodman1954some}
Goodman, L.~A. (1954).
\newblock Some practical techniques in serial number analysis.
\newblock {\em Journal of the American Statistical Association\/}~{\em
  49\/}(265), 97--112.

\bibitem[\protect\citeauthoryear{Grether}{Grether}{1981}]{grether1981financial}
Grether, D. (1981).
\newblock Financial incentive effects and individual decision-making.
\newblock Manuscript No 401, Division of the Humanities and Social Sciences,
  California Institute of Technology.

\bibitem[\protect\citeauthoryear{Heifetz and Neeman}{Heifetz and
  Neeman}{2006}]{heifetz2006generic}
Heifetz, A. and Z.~Neeman (2006).
\newblock On the generic (im)possibility of full surplus extraction in
  mechanism design.
\newblock {\em Econometrica\/}~{\em 74\/}(1), 213--233.

\bibitem[\protect\citeauthoryear{Karni}{Karni}{2009}]{karni2009mechanism}
Karni, E. (2009).
\newblock A mechanism for eliciting probabilities.
\newblock {\em Econometrica\/}~{\em 77\/}(2), 603--606.

\bibitem[\protect\citeauthoryear{Kechris}{Kechris}{1995}]{kechris1995classical}
Kechris, A.~S. (1995).
\newblock {\em Classical Descriptive Set Theory}.
\newblock Springer-Verlag.

\bibitem[\protect\citeauthoryear{Lambert, Pennock, and Shoham}{Lambert
  et~al.}{2008}]{lambert2008eliciting}
Lambert, N.~S., D.~M. Pennock, and Y.~Shoham (2008).
\newblock Eliciting properties of probability distributions.
\newblock In {\em Proceedings of the 9th ACM Conference on Electronic
  Commerce}, pp.\  129--138.

\bibitem[\protect\citeauthoryear{Le~Cam}{Le~Cam}{1996}]{le1996comparison}
Le~Cam, L. (1996).
\newblock Comparison of experiments: {A} short review.
\newblock In T.~S. Ferguson, L.~S. Shapley, and J.~B. MacQueen (Eds.), {\em
  Statistics, Probability and Game Theory}, pp.\  127--138. Institute of
  Mathematical Statistics.

\bibitem[\protect\citeauthoryear{Lehrer and Shmaya}{Lehrer and
  Shmaya}{2008}]{LehrerShmaya2008}
Lehrer, E. and E.~Shmaya (2008).
\newblock Two remarks on {{Blackwell}}'s theorem.
\newblock {\em Journal of applied probability\/}~{\em 45\/}(2), 580--586.

\bibitem[\protect\citeauthoryear{Lopomo, Rigotti, and Shannon}{Lopomo
  et~al.}{2019}]{lopomo2019detectability}
Lopomo, G., L.~Rigotti, and C.~Shannon (2019).
\newblock Detectability, duality, and surplus extraction.
\newblock Working paper (arXiv 1905.12788).

\bibitem[\protect\citeauthoryear{Matheson and Winkler}{Matheson and
  Winkler}{1976}]{matheson1976scoring}
Matheson, J.~E. and R.~L. Winkler (1976).
\newblock Scoring rules for continuous probability distributions.
\newblock {\em Management Science\/}~{\em 22\/}(10), 1087--1096.

\bibitem[\protect\citeauthoryear{McAfee and Reny}{McAfee and
  Reny}{1992}]{mcafee1992correlated}
McAfee, R.~P. and P.~J. Reny (1992).
\newblock Correlated information and mechanism design.
\newblock {\em Econometrica\/}~{\em 60\/}(2), 395--421.

\bibitem[\protect\citeauthoryear{McCarthy}{McCarthy}{1956}]{mccarthy1956measures}
McCarthy, J. (1956).
\newblock Measures of the value of information.
\newblock {\em Proceedings of the National Academy of Sciences\/}~{\em
  42\/}(9), 654--655.

\bibitem[\protect\citeauthoryear{Milgrom and Segal}{Milgrom and
  Segal}{2002}]{milgrom2002envelope}
Milgrom, P. and I.~Segal (2002).
\newblock Envelope theorems for arbitrary choice sets.
\newblock {\em Econometrica\/}~{\em 70\/}(2), 583--601.

\bibitem[\protect\citeauthoryear{Miller, Resnick, and Zeckhauser}{Miller
  et~al.}{2005}]{miller2005eliciting}
Miller, N., P.~Resnick, and R.~Zeckhauser (2005).
\newblock Eliciting informative feedback: {T}he peer-prediction method.
\newblock {\em Management Science\/}~{\em 51\/}(9), 1359--1373.

\bibitem[\protect\citeauthoryear{Neeman}{Neeman}{2004}]{neeman2004relevance}
Neeman, Z. (2004).
\newblock The relevance of private information in mechanism design.
\newblock {\em Journal of Economic theory\/}~{\em 117\/}(1), 55--77.

\bibitem[\protect\citeauthoryear{Olszewski and Sandroni}{Olszewski and
  Sandroni}{2008}]{olszewski2008}
Olszewski, W. and A.~Sandroni (2008).
\newblock Manipulability of future-independent tests.
\newblock {\em Econometrica\/}~{\em 76\/}(6), 1437--1466.

\bibitem[\protect\citeauthoryear{Rahman}{Rahman}{2012}]{rahman2012surplus}
Rahman, D. (2012).
\newblock Surplus extraction on arbitrary type spaces.
\newblock Working paper.

\bibitem[\protect\citeauthoryear{Roth and Malouf}{Roth and
  Malouf}{1979}]{roth1979game}
Roth, A.~E. and M.~W. Malouf (1979).
\newblock Game-theoretic models and the role of information in bargaining.
\newblock {\em Psychological Review\/}~{\em 86\/}(6), 574--594.

\bibitem[\protect\citeauthoryear{Savage}{Savage}{1971}]{savage1971elicitation}
Savage, L.~J. (1971).
\newblock Elicitation of personal probabilities and expectations.
\newblock {\em Journal of the American Statistical Association\/}~{\em
  66\/}(336), 783--801.

\bibitem[\protect\citeauthoryear{Shmaya}{Shmaya}{2008}]{shmaya2008}
Shmaya, E. (2008).
\newblock Many inspections are manipulable.
\newblock {\em Theoretical Economics\/}~{\em 3\/}(3), 367--382.

\bibitem[\protect\citeauthoryear{Strulovici}{Strulovici}{2017}]{strulovici2017impossibility}
Strulovici, B. (2017).
\newblock On the impossibility of full surplus extraction.
\newblock Working paper.

\bibitem[\protect\citeauthoryear{Teicher}{Teicher}{1961}]{teicher1961identifiability}
Teicher, H. (1961).
\newblock Identifiability of mixtures.
\newblock {\em The Annals of Mathematical Statistics\/}~{\em 32\/}(1),
  244--248.

\end{thebibliography}

\end{document}